\documentclass[12pt,aps,a4paper,showpacs,eqsecnum,amsmath
               ,amssymb,nofootinbib,superscriptaddress
               ,notitlepage]{revtex4}

\usepackage{bbm,amsthm}
\usepackage{graphicx}
\usepackage{tikz}

\newcommand{\dimension}{d}
\newcommand{\field}{C}

\newcommand{\norm}[1]{\ensuremath{\left| \! \left| {#1} \right| \! \right| }}  

\newcommand{\card}{\operatorname{card}} 
\newcommand{\supp}{\operatorname{supp}} 
\newcommand{\diag}{\operatorname{diag}} 

\renewcommand{\Re}{\operatorname{Re}} 
\renewcommand{\Im}{\operatorname{Im}} 

\newcommand{\I}{\mathrm{i}}        
\newcommand{\D}{\mathrm{d}}        
\newcommand{\E}{\mathrm{e}}        

\newcommand{\MV}[1]{\operatorname{E}\{{#1}\}}	
\newcommand{\Cov}[1]{\operatorname{Cov}\{{#1}\}}

\newcommand{\zvec}{\ensuremath{\mathbf z}}
\newcommand{\yvec}{\ensuremath{\mathbf y}}
\newcommand{\xvec}{\ensuremath{\mathbf x}}
\newcommand{\wvec}{\ensuremath{\mathbf w}}
\newcommand{\vvec}{\ensuremath{\mathbf v}}
\newcommand{\uvec}{\ensuremath{\mathbf u}}
\newcommand{\svec}{\ensuremath{\mathbf s}}
\newcommand{\rvec}{\ensuremath{\mathbf r}}
\newcommand{\qvec}{\ensuremath{\mathbf q}}
\newcommand{\pvec}{\ensuremath{\mathbf p}}
\newcommand{\kvec}{\ensuremath{\mathbf k}}
\newcommand{\fvec}{\ensuremath{\mathbf f}}
\newcommand{\evec}{\ensuremath{\mathbf e}}
\newcommand{\bvec}{\ensuremath{\mathbf b}}
\newcommand{\avec}{\ensuremath{\mathbf a}}

\newcommand{\Nullvec}{\ensuremath{\mathbf 0}}

\numberwithin{equation}{section}
\newtheorem*{theoremGoLo}{Theorem \cite{GoLo13}}
\newtheorem*{theoremHB}{Theorem \cite{Hackbusch93}}

\begin{document}
\renewcommand{\theequation}{\thesection.\arabic{equation}} 

\preprint{}
\title{Norm-minimized Scattering Data from Intensity Spectra}

\author{Alexander Seel}
\affiliation{%
Center for Sensorsystems, University of Siegen, 57068 Siegen, Germany
}

\author{Arman Davtyan}
\affiliation{%
Faculty of Science and Engineering, University of Siegen, 57068 Siegen, Germany
}

\author{Ullrich Pietsch}
\affiliation{%
Faculty of Science and Engineering, University of Siegen, 57068 Siegen, Germany
}

\author{Otmar Loffeld}
\affiliation{%
Center for Sensorsystems, University of Siegen, 57068 Siegen, Germany
}

\date{\today}

\begin{abstract}
We apply the $\ell_1$~minimizing technique of compressive sensing 
(CS) to non-linear quadratic observations. For the example of 
coherent $X$\!-ray scattering we provide the formulae for a Kalman 
filter approach to quadratic CS and show how to reconstruct 
the scattering data from their spatial intensity distribution.
\end{abstract}

\pacs{02.30.Zz, 02.50.-r, 07.85.-m}


\maketitle
\section{Introduction}
\enlargethispage{1\baselineskip}
The rapidly growing Si technology in semiconductor electronics 
\cite{BBV07,BoPc08} opens the possibility to grow III-V inorganic 
nanowires, such as GaAs, InAs or InP, which were supposed to 
have potential \cite{CuLi01,DHCWL01} to become building blocks in a 
variety of nanowire-based nanoelectronic devices, for example in 
nanolaser sources \cite{HMFYWWRY01} or nanoelectronics \cite{GLWSL02}. 
Such epitaxially grown nanowires are repeating the crystal orientation 
of the substrate and usually grow in Wurtzite (WZ) or Zinc-Blende (ZB) 
structure differing in the stacking sequence $AB\!AB\!AB$ and 
$ABC\!ABC\!ABC$ respectively of the atomic bilayers.
Theoretical predictions on the electronic properties \cite{AkNaIt15}
of these nanowires show that stacking sequences with WZ and ZB 
segments considerably differ in the conductivity.
However, during the nanowire growth stacking faults, the mixing 
of ZB and WZ segments takes place, and twin defects \cite{BCBJSGP13} 
appear. As these defects have their own impact on the conductivity and 
band structure there is great interest in knowing the exact stacking 
sequence which can be studied by e.g.\ Transmission Electron 
Microscopy \cite{ShDoWa12}. 
But, as this is a destructive method, it is impossible to use the 
nanowire after the structural studies. Nowadays the $3$rd generation 
synchrotron sources and rapidly developing focusing devices like 
Fresnel Zone Plates opens new fields of non-destructive \mbox{$X$\!-ray} 
imaging. For example in the Coherent $X$\!-ray Diffraction Imaging 
experiments~(CXDI) an isolated nanoobject is illuminated with 
coherent $X$\!-ray radiation and the scattered intensity is measured 
by a 2D detector \cite{ChNu10,MIRM15} under the Bragg condition. 
The diffraction patterns are structure-specific and encode the 
information about the electron density of the sample and thus 
the stacking sequence of the atomic bilayers formally by Fourier 
transform. 
However, because of sensor physics the phase information is lost in 
CXDI measurement since the measured intensity pattern is the modulus
square of the scattered $X$\!-rays and no inverse Fourier transform can 
be directly applied to recover the stacking sequence. 
{The classical approach with the Patterson function
\cite{Patterson34} fails as the number of expected randomly distributed
stacking faults is too high. Other} inversion algorithms 
\cite{GeSa72,Fienup78} could be used instead to reconstruct the lost 
phase information: Although dual space iterative algorithms 
\cite{Fienup82} have been shown to converge under specific conditions 
\cite{RoHa09} there still remain some convergence problems for a number 
of cases when not enough preliminary information regarding the 
structure of the object is previously known \cite{DBLP16}.
Indeed ptycho\-graphy type of experiments \cite{PGTJHM14,HCOATY15} 
were suggested to determine at least the relative phases of the 
bilayers by interfering adjacent Bragg reflections. But this type 
of experiments are more difficult to realize, as it requires higher 
stability in comparison to conventional CXDI and the longer 
measurement time can, however, influence the sample \cite{FEKG09}.

Due to the principal loss in phase information the scattering 
data can be considered to be undersampled and algorithms tailored 
to this lack of information, like basis pursuit approaches 
\cite{ChDoSa01} realized by minimizing the $\ell_1$ norm of 
sparse vectors in an appropriate basis {or phase 
retrieval via matrix completion \cite{CESV13,ChMoPa10,FoRaWa11}}, could be 
tested for reconstruction from conventional CXDI measurements: 
{As this data are recorded without structured 
illumination for applying matrix completion we focus on} utilizing 
the $\ell_1$ minimizing technique to reconstruct undersampled sparse 
signals \cite{CaRoTa06,Donoho06} {as vectors} from 
linear mappings between the signal and the observations. For an 
overview of this so called compressive sensing (CS) see the 
textbook \cite{FoRa13}. 
The method of CS aims to a signal's sparse support which can be 
estimated by e.g.\ Kalman filtering \cite{Vaswani08}.
As the underlying filter formulas also relate observations and 
signals to be reconstructed by linear mappings this technique 
meets CS and was also used to explicitly minimize the $\ell_1$ 
norm by so called pseudo measurements \cite{KCHGRS10}. For an 
example see the reconstruction from a random sample of Fourier 
coefficients \cite{LoEsCo15}. Even for 
non-linear mappings of signals Kalman filtering applies by using 
Jacobians rather than constant sensing matrices. These extended 
Kalman filters (EKF) are used for e.g.\ tracking issues 
\cite{JuUh97} or robotics \cite{Correll14} and match the sensing 
problem of the quadratic nonlinearities in the spatial intensity 
distribution of CXDI.

{The article is organized as follows:
In Section~\ref{intensity} we setup the observation model for 
modulus-squared amplitudes of intensity distributions in coherent 
$X$\!-ray scattering and point out the relation to the $\ell_1$ minimization.
In Section~\ref{kalmanmodel} we give a brief overview of the linear 
Kalman filter model and show how to incorporate the complex 
non-analytic $\ell_1$ norm as a linearized observation to meet the 
minimizing strategy of CS. In this framework we prove a convergence 
concept for the $\ell_1$ minimization in the reconstruction scheme.
In Section~\ref{linobsmod} we apply our findings to 1D and 2D 
Fourier data and remark in Section~\ref{outlook} 
on future considerations.}

\section{Motivation}\label{intensity}

\subsection{Intensity Spectra in Coherent Scattering Experiments}
Exposing crystals to non-destructive coherent $X$\!-ray radiation in 
$d$ dimensions the amplitude of the elastically scattered radiation 
is proportional to the Fourier transform \cite{Warren90}
\begin{equation} \label{signalout}
S(\qvec)=\int\!\!\D^\dimension x\,\E^{\I\langle\qvec|\rvec\rangle} b(\rvec)
\end{equation}
of the electron density $b(\rvec)$, where 
the vector $\qvec:=\kvec_{\mathrm{out}}-\kvec_{\mathrm{in}}$ is a 
parametrization of the direction where radiation is detected and 
(multiplied by Planck's constant) also describes the momentum transfer
in a kinematical scattering theory; 
$\kvec_{\mathrm{in}}$ is the incident direction whereas 
$\kvec_{\mathrm{out}}$ represents the outgoing direction for radiation 
with wavevectors $\norm{\kvec_{\mathrm{out}}}=
\norm{\kvec_{\mathrm{in}}}=\frac{2\pi}{\lambda}$
of wavelength $\lambda$:

\begin{center}
\begin{tikzpicture}[scale=0.63]
\foreach \t in {.2}
{
 \foreach \s in {1/7}
  {
  \draw[join=round] (-4.5*\t,-3.5*\t) -- ++(6*\t,-\t) -- ++(3*\t,2*\t) 
    -- ++(0,6*\t) -- ++(-6*\t,\t) -- ++(-3*\t,-2*\t) -- cycle;
  \draw[cap=round,join=round] (1.5*\t,-4.5*\t) -- ++(0,{6*\t-2*6*\s*\t}) 
    -- ++(-6*\s*\t,\s*\t) -- ++(0,6*\s*\t) -- ++(-2*6*\s*\t,2*\s*\t) 
    -- ++(0,6*\s*\t) -- ++({-6*\t+3*\s*6*\t},{\t-3*\s*\t});
  \draw[cap=round,join=round] (1.5*\t,{1.5*\t-2*6*\s*\t}) -- ++(3*\s*\t,2*\s*\t) 
    -- ++(0,6*\s*\t) -- ++(3*\s*\t,2*\s*\t) -- ++(0,6*\s*\t) 
    -- ++({3*\t-2*\s*3*\t},{2*\t-2*\s*2*\t});
  \draw[cap=round,join=round] ({1.5*\t-3*\s*6*\t)},{1.5*\t+3*\s*\t)}) 
    -- ++(3*\s*\t,2*\s*\t) -- ++(2*6*\s*\t,-2*\s*\t) -- ++(3*\s*\t,2*\s*\t) 
    -- ++(6*\s*\t,-\s*\t);
  \draw[cap=round,join=round] ({1.5*\t-3*\s*6*\t)},{1.5*\t-3*\s*\t}) 
    -- ++(3*\s*\t,2*\s*\t) -- ++(3*\s*6*\t,-3*\s*\t);
  \draw[cap=round,join=round] ({1.5*\t-6*\s*\t},{1.5*\t-11*\s*\t}) 
    -- ++(3*\s*\t,2*\s*\t) -- ++(0,6*2*\s*\t);
  \draw[cap=round,join=round] ({1.5*\t-3*\s*\t},{1.5*\t-9*\s*\t}) 
    -- ++(6*\s*\t,-\s*\t);
  \draw[cap=round,join=round] ({1.5*\t-6*\s*\t},{1.5*\t-5*\s*\t}) 
    -- ++(2*\s*3*\t,2*\s*2*\t) -- ++(0,6*\s*\t); 
  \draw[cap=round,join=round] ({1.5*\t-15*\s*\t},{1.5*\t+5*\s*\t}) 
    -- ++ (0,-6*\s*\t);
  \draw[cap=round,join=round] ({1.5*\t+6*\s*\t},{1.5*\t-2*\s*\t}) 
    -- ++(-6*\s*\t,\s*\t);
 }
 \foreach \a/\b in {4/75}
 {
  \draw[lightgray] (0,0) circle[radius=\a];
  \draw[lightgray] (\a,0) -- ++(0.7*\a,0);
  \draw[lightgray] (-\a,0) -- ++(-0.85*\a,0);
  \draw[very thick,-stealth,red,cap=round] (3*\t,0) -- (\a,0);
  \draw[very thick,-stealth,red,cap=round] (\b:3*\t) -- (\b:\a);
  \draw[very thick,-stealth,cap=round] (0:\a) -- (\b:\a);
  \node[red] at ({\a/2},-.4) {$\kvec_\mathrm{in}$};
  \node[red,rotate=\b] at ({cos(\b)*\a/2-.4},{sin(\b)*\a/2}) {$\kvec_\mathrm{out}$};
  \node at ({\a*(cos(\b)+1)/2-.35},{\a/2*sin(\b)-.2}) {$\qvec$};
  \foreach \w/\u in {180/0.25}
  {
   \draw[ultra thick,color=red!25,stealth-] (\w:{\a*0.6}) -- (\w:{1.6*\a});
   \foreach \c in {-2,...,3}
   \draw[cap=round,color=red!50] ({(\a+\c*\u)*cos(\w)+sin(\w)},{(\a+\c*\u)*sin(\w)-cos(\w)}) 
    -- ++ ({\w+90}:2);
  }
  \foreach \w/\u in {140/0.25,-142/0.25,-33/0.25}
  {
   \draw[very thick,color=black!15,-stealth,cap=round] (\w:{\a*0.4}) -- (\w:{1.4*\a});
   \foreach \c in {-3,...,2}
   \draw[cap=round,lightgray] ({(\a+\c*\u)*cos(\w)+sin(\w)},{(\a+\c*\u)*sin(\w)-cos(\w)}) 
    -- ++ ({\w+90}:2);
  }
  \foreach \w in {-110,-71} 
  {
   \draw[cap=round,thick] ({0.9957*\a*cos(\w)-0.5*sin(\w)},{0.997*\a*sin(\w)+0.5*cos(\w)}) 
    -- ({0.997*\a*cos(\w)+0.5*sin(\w)},{0.997*\a*sin(\w)-0.5*cos(\w)});
   \draw[cap=round,thick,fill=white] plot[smooth,tension=2] 
     coordinates {({\a*cos(\w)-0.5*sin(\w)},{\a*sin(\w)+0.5*cos(\w)}) 
     ({(\a+.8)*cos(\w)},{(\a+.8)*sin(\w)}) 
     ({\a*cos(\w)+0.5*sin(\w)},{\a*sin(\w)-0.5*cos(\w)})};
   \draw[cap=round,thick] plot[smooth,tension=.7] 
     coordinates {({(\a+.8)*cos(\w)},{(\a+.8)*sin(\w)}) 
     ({(\a+1.25)*cos(\w+2)},{(\a+1.25)*sin(\w+2)}) 
     ({(\a+1.5)*cos(\w+8.5)},{(\a+1.5)*sin(\w+8.5)})
     ({(\a+2)*cos(\w+10)},{(\a+2)*sin(\w+10)})};
  }
  \foreach \w in {40} 
  {
   \draw[cap=round,thick] ({0.9957*\a*cos(\w)-0.5*sin(\w)},{0.997*\a*sin(\w)+0.5*cos(\w)}) 
    -- ({0.997*\a*cos(\w)+0.5*sin(\w)},{0.997*\a*sin(\w)-0.5*cos(\w)});
   \draw[cap=round,thick,fill=white] plot[smooth,tension=2] 
     coordinates {({\a*cos(\w)-0.5*sin(\w)},{\a*sin(\w)+0.5*cos(\w)}) 
     ({(\a+.8)*cos(\w)},{(\a+.8)*sin(\w)}) 
     ({\a*cos(\w)+0.5*sin(\w)},{\a*sin(\w)-0.5*cos(\w)})};
   \draw[cap=round,thick] plot[smooth,tension=.7] 
     coordinates {({(\a+.8)*cos(\w)},{(\a+.8)*sin(\w)}) 
     ({(\a+1.25)*cos(\w-2)},{(\a+1.25)*sin(\w-2)}) 
     ({(\a+1.5)*cos(\w-8.5)},{(\a+1.5)*sin(\w-8.5)})
     ({(\a+2)*cos(\w-10)},{(\a+2)*sin(\w-10)})};
  }
  \foreach \w in {-33} 
  {
   \draw[cap=round,thick] ({0.9957*\a*cos(\w)-0.5*sin(\w)},{0.997*\a*sin(\w)+0.5*cos(\w)}) 
     -- ({0.997*\a*cos(\w)+0.5*sin(\w)},{0.997*\a*sin(\w)-0.5*cos(\w)});
   \draw[cap=round,thick,fill=white] plot[smooth,tension=2] 
     coordinates {({\a*cos(\w)-0.5*sin(\w)},{\a*sin(\w)+0.5*cos(\w)}) 
     ({(\a+.8)*cos(\w)},{(\a+.8)*sin(\w)}) 
     ({\a*cos(\w)+0.5*sin(\w)},{\a*sin(\w)-0.5*cos(\w)})};
   \draw[cap=round,thick] plot[smooth,tension=.7] 
     coordinates {({(\a+.8)*cos(\w)},{(\a+.8)*sin(\w)}) 
     ({(\a+1.25)*cos(\w+2)},{(\a+1.25)*sin(\w+2)}) 
     ({(\a+1.5)*cos(\w+8.5)},{(\a+1.5)*sin(\w+8.5)})
     ({(\a+2)*cos(\w+10)},{(\a+2)*sin(\w+10)})};
   \draw[fill=white]  ({(\a+1.25)*cos(\w+2)},{(\a+1.25)*sin(\w+2)}) circle (.1);
   \foreach \c in {50,90,130,170,-28,-68,-108}
    {
     \draw ({(\a+1.25)*cos(\w+2)+0.15*cos(\c+\w-152)},{(\a+1.25)*sin(\w+2)+0.15*sin(\c+\w-152)}) 
       -- ++(\c+\w-152:.24);
    }
  }
  \node[red!25] at (-1.35*\a,-0.4) {$\kvec_\mathrm{in}$};
 }
}
\draw[fill] (-.22,-.25) circle[radius=.1];
\draw (-.22,-.25) -- (-.5,-1.4);
\node at (-.3,-1.9) {$b(\rvec)$};
\end{tikzpicture}
%
\end{center}
\newpage

Following standard textbooks, e.g.\ \cite{Warren90}, one deals with 
two sets of basis vectors to characterize the scattering. The spatial 
vectors are represented in the basis 
$\big\{\avec_1,\ldots,\avec_\dimension\big\}$ 
of the grid whereas the reciprocal basis 
$\big\{\bvec_1,\ldots,\bvec_\dimension\big\}$ 
is used for wave vectors $\qvec,\kvec_\mathrm{in},\kvec_\mathrm{out}$ 
relying on the normalization
\begin{equation}\label{normalization}
\langle\avec_j|\bvec_k\rangle=2\pi\delta_{jk} 
\quad , \quad j,k=1,\ldots,\dimension\quad .
\end{equation}

As the units of lengths and wave vectors are carried by the basis 
sets the scalar products read with dimensionless factors $y_j$ and 
$\kappa_k$
\begin{equation} \label{lincombi}
\rvec=\sum_{j=1}^\dimension y_j\avec_j \quad,\quad 
\qvec=\sum_{k=1}^\dimension \kappa_k\bvec_k
\quad ,\quad
\langle\qvec|\rvec\rangle=2\pi\sum_{j=1}^\dimension y_j \kappa_j \quad .
\end{equation}

Restricting to a finite grid with $n_1,n_2\ldots,n_\dimension$ 
lattice sites in each direction with periodic boundary conditions the 
possible grid positions and wave vectors allowing for a discrete
Fourier transform\footnote{
In solid state physics the reciprocal vectors $\qvec$ according to
(\ref{lincombi}) and (\ref{1stBrillouin}) for the $\kappa_j$ are said 
to be from the $1$st Brillouin zone which is a fragmentation of 
the elementary cell spanned by $\{\bvec_1,\ldots,\bvec_\dimension\}$.
}
read for all $j=1,2,\ldots,\dimension$
\begin{equation}\label{1stBrillouin}
y_j\in\big\{0,1,\ldots,n_j-1\big\} \quad ,\quad 
\kappa_j\in\bigg\{\frac0{n_j},\frac{1}{n_j},\ldots,\frac{n_j-1}{n_j}\bigg\}
\quad .
\end{equation}

Thus the scalar product separates into the $d$ spatial dimensions 
according to
\begin{equation}
\langle\qvec|\rvec\rangle=\sum_{j=1}^\dimension\frac{2\pi k_j r_j}{n_j}
\quad ,\quad
k_j,r_j\in\big\{0,1,\ldots,n_j-1\big\}
\end{equation}
where the $k_j$ are related to the lattice positions in direction 
of the lattice constant $\avec_j$ and $r_j$ refer to a discretized 
wave vector in direction of the corresponding reciprocal basis $\bvec_j$.
For  example a $d=2$ dimensional regular hexagonal lattice 
$\{\avec_1,\avec_2\}$ encloses $120^\circ$-angles: 

\begin{center}
\vspace*{1ex}
\begin{tikzpicture}[scale=0.85]
\foreach \m in {6}
{
\draw[lightgray,join=round] (0,0) -- ++(\m,0) -- ++(60:3) 
  -- ++(120:3) -- ++(-\m,0) -- ++(-120:3) -- cycle;
\draw[lightgray,cap=round] (-.5,{sqrt(3)/2}) -- ++(\m+1,0);
\draw[lightgray,cap=round] (-1,{sqrt(3)}) -- ++(\m+2,0);
\draw[lightgray,cap=round] (-1.5,{3*sqrt(3)/2}) -- ++(\m+3,0);
\draw[lightgray,cap=round] (-1,{2*sqrt(3)}) -- ++(\m+2,0);
\draw[lightgray,cap=round] (-.5,{5*sqrt(3)/2}) -- ++(\m+1,0);
\foreach \t in {3,...,\m}
 {
  \draw[lightgray,cap=round] (\t-3,0) -- ++(60:6);
  \draw[lightgray,cap=round] (\t-3,{3*sqrt(3)}) -- ++(-60:6);
 }
 \draw[lightgray,cap=round] (-.5,{sqrt(3)/2}) -- ++(60:5);
 \draw[lightgray,cap=round] (-1,{sqrt(3)}) -- ++(60:4);
 \draw[lightgray,cap=round] (-.5,{5*sqrt(3)/2}) -- ++(-60:5);
 \draw[lightgray,cap=round] (-1,{2*sqrt(3)}) -- ++(-60:4);
 \draw[lightgray,cap=round] (\m+.5,{sqrt(3)/2}) -- ++(120:5);
 \draw[lightgray,cap=round] (\m+1,{sqrt(3)}) -- ++(120:4);
 \draw[lightgray,cap=round] (\m+1,{2*sqrt(3)}) -- ++(-120:4);
 \draw[lightgray,cap=round] (\m+.5,{5*sqrt(3)/2}) -- ++(-120:5);
\draw[ultra thick,join=round] (1.5,{sqrt(3)/2}) -- ++(\m-2,0) -- ++(120:3)
  -- ++(180:\m-2) -- cycle;
\coordinate (A) at (2,{sqrt(3)});
\draw[fill=black!12] (A) -- ++(1,0) -- ++(120:1) -- ++(180:1) -- cycle;
\draw[very thick,-stealth,cap=round] (A) -- ++(1,0);
\draw[very thick,-stealth,cap=round] (A) -- ++(120:1);
\node at (2.52,1.47) {$\avec_1$};
\node at (1.55,1.98) {$\avec_2$};
}
\begin{scope}[xshift=10cm,yshift=.4cm]
\foreach \m/\s in {5/1.6}
{
\draw[lightgray,join=round] (0,0) -- ++(-30:\s) -- ++(30:{2*\s}) -- ++(0,{(\m-3)*\s}) 
  -- ++(150:\s) -- ++(-150:{2*\s}) -- ++(0,{(3-\m)*\s}) -- cycle;
\draw[lightgray,cap=round] (0,0)-- ++(30:{3*\s});
\draw[lightgray,cap=round] (0,\s)-- ++(30:{3*\s});
\draw[lightgray,cap=round] (0,\s)-- ++(-30:{2*\s});
\draw[lightgray,cap=round] (0,{2*\s})-- ++(-30:{3*\s});
\draw[lightgray,cap=round] ({\s*sqrt(3)/2},{2.5*\s})-- ++(-30:{2*\s});
\draw[lightgray,cap=round] ({\s*sqrt(3)/2},{-.5*\s}) -- ++(0,{3*\s});
\draw[lightgray,cap=round] ({\s*sqrt(3)},0) -- ++(0,{3*\s});
\coordinate (B) at ({\s*sqrt(3)/2},{.5*\s});
\draw[fill=black!12] (B) -- ++(30:\s) -- ++(0,\s) -- ++(-150:\s) -- cycle;
\draw[very thick,-stealth,cap=round] (B) -- ++(30:\s);
\draw[very thick,-stealth,cap=round] (B) -- ++(90:\s);
\foreach \a in {1,2}
{
  \draw[densely dashed] ({\s*sqrt(3)/2},{.5*\s+\a*\s/3}) -- ++(30:\s);
}
\foreach \a in {1,2,3}
 {
  \draw[densely dashed] ({\s*sqrt(3)/2+\s*sqrt(3)*\a/2/4},{.5*\s+.5*\s*\a/4}) -- ++(0,\s);
 }
\node at (0.65*\s,1.05*\s) {$\bvec_2$};
\node at (1.42*\s,.55*\s) {$\bvec_1$};
}
\end{scope}
\end{tikzpicture}
\end{center}

The bold parallelogram 
represents periodic boundary conditions with e.g.\ $n_1=4$ and 
\mbox{$n_2=3$} fragmenting the 1st Brillouin zone, the elementary cell 
spanned by $\{\bvec_1,\bvec_2\}$ in reciprocal space, into a subgrid
according to (\ref{1stBrillouin}). With the normalization 
(\ref{normalization}) the reciprocal lattice encloses angles of 
$60^\circ$ with a fixed orientation with respect to $\{\avec_1,\avec_2\}$.

\subsection{Setting}
Of course for the application the formulas are only needed up to 
$d=3$~dimensions: Nanowires  can be characterized 
by the stacking sequence of atomic bilayers which are shifted laterally 
and vertically with respect to each other.
In coherent $X$\!-ray diffraction summing up all scattered radiation 
from a certain bilayer perpendicular to the growth direction
$\avec_3$ yields 
a complex scattering amplitude $x_k\in\mathbb{C}$ associated with 
the $k$th bilayer spanned by $\big\{\avec_1,\avec_2\big\}$. 
Due to the hexagonal lattice and with respect to an arbitrary 
reference bilayer there are three different phase factors 
\begin{equation}\label{phasesc}
\Big\{1 \;,\; \exp\Big(\frac{2\pi\I}{3}\big(2\kappa_1+\kappa_2\big)\Big) 
\;,\; \exp\Big(\frac{2\pi\I}{3}\big(\kappa_1+2\kappa_2\big)\Big) \Big\}
\end{equation}
for the amplitudes left \cite{Warren90,BHKSSS11}. Wave vectors $\qvec$
sensitive to the arrangement of the bilayers are selected
from directions related to the Bragg condition by
$\kappa_1-\kappa_2\not=3N$  
with $N,\kappa_1,\kappa_2\in\mathbb{Z}$ yielding
\mbox{$\big\{1,\exp\big(-\frac{2\pi\I}{3}\big),
\exp\big(\frac{2\pi\I}{3}\big)\big\}$},
which directly relates to the three possible lateral shifts.
With respect to the experimental setup \cite{DBLP16} this is satisfied
by $\kappa_1=0$ and $\kappa_2=1$. Recovering the stacking 
sequence of this relative phase factors by varying $\kappa_3$ 
directly reveals the Zinc-Blende or Wurtzite structure of the 
wire along with their stacking faults.

Carrying out the 
remaining summation over all equidistant bilayers in growth 
direction~$\avec_3$ mathematically corresponds 
to the 1D discrete Fourier transform 
\begin{equation}
S_r=\sum_{k=0}^{n-1}\E^{\frac{2\pi\I k r}{n}} x_k 
\quad , \quad
r=0,1,\ldots,n-1
\end{equation}
of the periodically continued complex scattering amplitudes~$x_k$
combined to the vector $\xvec$. Each non-zero entry represents 
one bilayer and $r$ refers to the discretized component $\kappa_3$ 
of the $\qvec$ vector in the reciprocal basis corresponding to the 
growth direction.
As the outcome of the detector is the measured intensity rather
than the amplitude the observation is the squared signal
\begin{equation} \label{signal}
\big|S_r\big|^2 = \langle\xvec|\!|T_r|\xvec\rangle
\quad , \quad r=0,1,\ldots,n-1
\end{equation}
where the Fourier coefficients build a hermitian Toeplitz matrix
\begin{equation}\label{toeplitz}
T_r:=\big(\E^{-\frac{2\pi\I r}{n}(p-q)}\big)_{pq}\in\mathbb{C}^{n\times n}
\quad , \quad r=0,1,\ldots,n-1
\end{equation}
showing orthogonality $T_r T_s = n \delta_{rs} T_r$ with respect 
to multiplication and a completeness property
$T_0+T_1+\ldots +T_{n-1}=n\mathbbm{1}_n$ with respect to summation.
Clearly, due to the squaring the signal (\ref{signal}) is invariant 
both under the transformation $\xvec\to\E^{\I\phi}\xvec$ with any 
global phase $\phi\in\mathbb{R}$ and the reflection $x_k\to \overline{x}_{n-1-k}$,  
$k=0,\ldots,n-1$ on the lattice. As the spatial directions 
exhibit periodic boundary conditions the signal also remains invariant 
under discrete translations $x_k\to x_{k+q}$, $q\in\mathbb{Z}$ on the grid.

On the contrary, setting the $\qvec$ component $\kappa_3$ corresponding 
to the growth direction $\avec_3$ in the reciprocal basis to zero, 
(\ref{signalout}) examines the 2D structure of all the $M$ bilayers 
simultaneously. If they are assumed to be identical without 
lateral\footnote{
However, in the hexagonal lattice there are three possible 
lateral shifts yielding additional phase factors in 
(\ref{Fourier2D}) for each layer. So carrying out the complete
Fourier transform of all bilayers  yields a random sequence 
$c_1+c_2+\ldots+c_M$ with
\begin{equation}\label{phases}
c_j\in\bigg\{1 \;,\; \exp\bigg(\frac{2\pi\I}{3}\Big(\frac{2r_1}{n_1}+\frac{r_2}{n_2}\Big)\bigg) 
\;,\; \exp\bigg(\frac{2\pi\I}{3}\Big(\frac{r_1}{n_1}+\frac{2r_2}{n_2}\Big)\bigg) \bigg\}
\quad ,\quad j=1,\ldots,M
\end{equation}
rather than the constant $M$ in front of (\ref{Fourier2D}) and 
(\ref{FTsquare2D}). In the general case the $\qvec$ vectors
corresponding to the bilayers can be expressed by
\begin{math}
\qvec=\kappa_1 \bvec_1 + \kappa_2 \bvec_2
\end{math}. Then the phase factors (\ref{phases}) read
\begin{equation}
c_j\in\Big\{1 \;,\; \exp\Big(\frac{2\pi\I}{3}\big(2\kappa_1+\kappa_2\big)\Big) 
\;,\; \exp\Big(\frac{2\pi\I}{3}\big(\kappa_1+2\kappa_2\big)\Big) \Big\}
\quad ,\quad j=1,\ldots,M
\end{equation}
and their $\qvec$ dependence can be suppressed by choosing 
wave vectors  related to the Bragg condition
by $\kappa_1-\kappa_2=3N$ with $N,\kappa_1,\kappa_2\in\mathbb{Z}$.
One could be misled that this selects a certain amount of data 
points on the grid spanned by $\{\bvec_1,\bvec_2\}$ allowing for 
CS techniques. Switching to the discrete version
this reduces to one corner of the $1$st~Brillouin zone, i.e.\ to
the case of one single data point $r_1=r_2=0$ which is in fact 
too little for CS.
}
shifts the 
scattering amplitudes of the bilayer's lattice sites are 2D real 
data sets $X_{k_1 k_2}$ with the discrete Fourier transform 
\begin{equation} \label{Fourier2D}
S_{r_1 r_2} =M\cdot \sum_{k_1=0}^{n_1-1}\sum_{k_2=0}^{n_2-1}
\exp\Big({\frac{2\pi\I k_1 r_1}{n_1}}\Big) \cdot
\exp\Big({\frac{2\pi\I k_2 r_2}{n_2}}\Big)\cdot X_{k_1 k_2}
\end{equation}
for fixed $r_1=0,1,\ldots,n_1-1$ and $r_2=0,1,\ldots,n_2-1$
discretizing $\kappa_1$ and $\kappa_2$.
Squaring (\ref{Fourier2D}) yields Toeplitz matrices with Kronecker 
product structure (\ref{kronecker}) for the detected signal
\begin{equation} \label{FTsquare2D}
\big|S_{r_1 r_2}\big|^2 = M^2\cdot
\langle\xvec|\!| T_{r_1} \otimes T_{r_2}|\xvec\rangle 
\end{equation}
with multiple row and column indices $(p_1 p_2)$ and $(q_1 q_2)$ 
respectively referring to the ordinary rows $p_1,p_2$ and columns 
$q_1,q_2$ of the Toeplitz matrices (\ref{toeplitz}):
\begin{equation} \label{mapping2D1D}
\big(T_{r_1} \otimes T_{r_2}\big)_{(p_1 p_2),(q_1 q_2)} 
= \big(T_{r_1}\big)_{p_1 q_1} \big(T_{r_2}\big)_{p_2 q_2}
\quad , \quad \big(\xvec\big)_{(q_1 q_2)} := X_{q_1 q_2}
\end{equation}


{The scattering amplitudes in Fourier domain 
can be assigned with arbitrary phases. As this leaves the given intensity 
distribution $|S|^2$ unchanged an related infinite set of broadened 
scattering data is formed by convolution. But if the signal $|S|^2$
is assumed to be oversampled the original scattering data in this 
set appear to be sparse and can be recovered  by an $\ell_1$~minimization 
run on e.g.\ (\ref{signal}) or (\ref{FTsquare2D}) -- up to translations, 
reflections or global phases.}

\section{Kalman Filter-driven $\ell_1$ minimization}\label{kalmanmodel}

\subsection{Kalman Filter Equations}\label{kalmanmodeleq}
{The equations for the Kalman filter usually apply to 
vectors over the field of real numbers~$\mathbbm{R}$. We will examine 
next that the equations can also be extended to the field $\mathbbm{\field}$}.
Anticipating this in the classical state space approach \cite{Maybeck79a} 
the vectorial quantity $\xvec_k$ is supposed to evolve according to the 
linear evolution
\begin{subequations}\label{themodel}
\begin{align}
\label{themodellinear}
\xvec_{k+1} = &A_{k}\xvec_k +\uvec_k + G\wvec \in\mathbb{\field}^n \text{ stochastic}\\
&A_k\in \mathbb{\field}^{n\times n} \;\;,\;\; \uvec_k\in\mathbb{\field}^n \text{ deterministic}\\
&G\in \mathbb{\field}^{n\times r} \;\;,\;\; \wvec\in\mathbb{\field}^r \text{ stochastic}  
\end{align}
\end{subequations}
with a fixed matrix $G$ and a determined sequence of evolution 
matrices $A_k$. In the model above the quantity $\xvec_k$ is assumed 
to be only traceable indirectly by linear observations $\yvec_k$ which 
can be viewed as linear mappings with given sensing matrices $C_k$ by
\begin{subequations} \label{obsmod}
\begin{align}
\yvec_{k} = &C_k\xvec_k + \vvec\in\mathbb{\field}^m \text{ stochastic}\\
&C_k\in \mathbb{\field}^{m\times n} \;\;,\;\; \vvec\in\mathbb{\field}^m \text{ stochastic .} 
\end{align}
\end{subequations}

The stochastic behaviour of all the considered quantities is modelled 
{by zero-mean Gaussian distributions with positive definite
matrices $R$ and $Q$}. These covariance matrices and related mean values can be 
calculated from normalized Gaussian integrals denoted by expectation 
or mean values $\MV{\circ}$,
\begin{equation}
\begin{aligned}
\MV{\wvec}&=\Nullvec\in\mathbb{\field}^r\quad ,\;\;\\ 
\MV{\vvec}&=\Nullvec\in\mathbb{\field}^m\quad ,\;\; 
\end{aligned}
\begin{aligned}
\Cov{\wvec}&= \MV{|\wvec\rangle\langle\wvec|}=Q\in\mathbb{\field}^{r\times r} \\
\Cov{\vvec}&= \MV{|\vvec\rangle\langle\vvec|}=R\in\mathbb{\field}^{m\times m}
\end{aligned}
\end{equation}

The main idea behind Kalman filtering is to invert the observation 
model (\ref{obsmod}), where the estimation of the quantities $\xvec_k$ 
from the measurements $\yvec_k$ shows predictor-corrector structure 
\cite{Maybeck79a}: The prediction step relates the estimates 
$\xvec_{k+1}^-$ and $\xvec^+_k$ by extrapolation,
\begin{subequations}
\label{kalmanpredict}
\begin{align} 
\xvec_{k+1}^- &= A_k\xvec_k^+ + \uvec_k \quad ,\quad P_k^+ 
=\Cov{\xvec_k^+}\\ 
P_{k+1}^- &= A_k P_k^+ A_k^H + GQG^H \quad,
\end{align}
\end{subequations}
whereas the correction step updates the estimate $\xvec_{k+1}^-$ by 
relating it to the new measurement $\yvec_{k+1}$: As the underlying 
equations utilize conditional mean values the update can be formulated 
in terms of covariances
\begin{subequations}
\label{kalmanupdate}
\begin{align} 
\label{ccycle}
\big(P_{k+1}^+\big)^{-1} &= \big(P_{k+1}^-\big)^{-1} 
+ C_{k+1}^H R^{-1}C_{k+1}\\
\label{update}
\xvec_{k+1}^+ &= P_{k+1}^+ \big(P_{k+1}^-\big)^{-1} \xvec_{k+1}^- 
+ P_{k+1}^+ C_{k+1}^H R^{-1}\yvec_{k+1} \quad .
\end{align}
\end{subequations}

An alternative form of the covariance cycle in the correction step 
above reads in term of the explicit so called Kalman gain 
matrix $K_{k+1}$: 
\begin{subequations} \label{alternative}
\begin{align} 
K_{k+1}&=P_{k+1}^- C_{k+1}^H 
\big(C_{k+1}P_{k+1}^-C_{k+1}^H+R\big)^{-1}\\
P_{k+1}^+ &= \big(\mathbbm{1}_n-K_{k+1}C_{k+1}\big)P_{k+1}^-\\
\xvec_{k+1}^+ &= \xvec_{k+1}^- 
+ K_{k+1}\big(\yvec_{{k+1}}-C_{k+1}\xvec_{k+1}^-\big) \quad .
\end{align}
\end{subequations}

\subsection{$\ell_1$ minimization} \label{minimization}
The main problem in CS is to minimize $\norm{\xvec}_1$ subject to the 
constraint $\yvec=C\xvec$ for a fixed $\yvec$ and $C$. Because this
linear structure matches the evolution equation (\ref{themodellinear}) 
of Kalman filtering in Subsection~\ref{kalmanmodeleq} we want to follow 
\cite{KCHGRS10} using a pseudo measurement to iteratively minimize 
the $\ell_1$ norm: 
Here, the $k$th estimate of the state vector $\xvec$ can be associated 
by~$\xvec_k$ and the norm minimization is driven by the fixed constraint
$\yvec$ augmented by the lowered norm $\gamma\norm{\xvec_{k-1}}_1$ 
from the previous step by a factor of $0<\gamma\leq1$ as an additional 
observation.
The main issue of applying the Kalman filter to the $\ell_1$ minimization
procedure consists in suitably linearizing the non-analytic $\ell_1$ norm:
Solving
\begin{equation}
|z|^2 = \Re^2\Big(\frac{z \overline{z_0}}{|z_0|}\Big) 
+ \Im^2\Big(\frac{z \overline{z_0}}{|z_0|}\Big)
\end{equation}
for $|z|$ in the vicinity of $z_0\in\mathbb{C}$ yields for 
$\varphi:=\arg\big(\frac{z\overline{z_0}}{|z_0|}\big)$ the expression
\begin{equation} \label{modexpand}
|z| = \Re\Big(\frac{z \overline{z_0}}{|z_0|}\Big) \sqrt{1+\tan^2\varphi}
 \geq \Re\Big(\frac{z \overline{z_0}}{|z_0|}\Big) 
\quad .
\end{equation}

Thus the $\ell_1$ norm of any vector $\zvec\in\mathbb{C}^n$ can be 
linearized in the vicinity of $\zvec_0\in\mathbb{C}^n$ utilizing its
phase information {$\pvec\in\mathbb{C}^n$} by
\begin{equation}\label{normbound}
 \Re \langle\pvec|\zvec\rangle\leq \norm{\zvec}_1\quad ,\quad
\langle\pvec|=\langle\pvec|(\zvec_0):=
\bigg(\frac{(\overline{\zvec_0})_1}{|(\zvec_0)_1|},\ldots,
    \frac{(\overline{\zvec_0})_n}{|(\zvec_0)_n|}\bigg)\quad.
\end{equation}

{Due to vectors over the field $\mathbb{C}$ it is 
necessary to distinguish row vectors $\langle \pvec|$ from the usual
column vectors $|\pvec\rangle:=\pvec$ related by hermitian conjugation. 
The notation is based on complex scalar products and matrix 
multiplication, cf.\ Appendix~\ref{appvectors}.}

\subsection{Pseudo Measurements}
Involving the observation $\yvec=C\xvec\in\mathbb{C}^m$ as a constraint
like e.g.\ a Fourier transform the Kalman filter equations 
(\ref{kalmanpredict}) and (\ref{kalmanupdate})
read (for $A=G=\mathbbm{1}$, $\uvec_k=\Nullvec$) for the estimates 
{$\xvec_k:=\xvec_k^+$} to the vector $\xvec$
{involving a full prediction-correction step}
\begin{subequations}
\begin{align}
\label{ccycle2}
P_{k+1}^{-1} &= (P_k+Q)^{-1}+ C_{k+1}^H R^{-1} C_{k+1}\\
 \label{update2}
\xvec_{k+1} &= P_{k+1}P_k^{-1}\xvec_k + P_{k+1}C_{k+1}^H R^{-1}\yvec_{k+1}
\end{align}
\end{subequations}
with the initial values $P_0\in\mathbb{C}^{n\times n}$, 
$\xvec_0\in\mathbb{C}^n$ and the constant parameters 
$C\in\mathbb{C}^{m\times n}$, $R\in\mathbb{C}^{(m+1)\times n}$ and 
$Q\in\mathbb{C}^{n\times n}$ to reconstruct $\xvec=\lim_{k\to\infty}\xvec_k$ 
from a given $\yvec$ as a limiting value. The minimization of the $\ell_1$ 
norm is incorporated into the (pseudo) measurements by
\begin{equation}
\yvec_{k+1}=
\Bigg(\begin{matrix}\yvec\\\hline
   \gamma_k\norm{\xvec_k}_1\end{matrix}\Bigg)\in\mathbb{C}^{m+1}
\quad ,\quad
C_{k+1}= 
\Bigg(\begin{matrix}C\\ \hline \langle\pvec|(\xvec_k)\end{matrix}\Bigg) 
\in\mathbb{C}^{(m+1)\times n}\quad ,
\end{equation}
where the factor $0<\gamma_k\leq1$ adaptively lowers the current
$\ell_1$ norm solving for a corresponding estimate $\xvec_{k+1}$ in 
the next iteration step. For a solution to the linear constraint 
$C\in\mathbb{C}^{m\times n}$ as initial data can serve
\begin{equation}
\xvec_0 =
\begin{cases}
(C^H C)^{-1}C^H\yvec \text{ for } m>n\\
C^H (CC^H)^{-1}\yvec \text{ for } m\leq n
\end{cases}\quad.
\end{equation}

\subsection{Convergence Considerations} 
By inserting (\ref{ccycle2}) into (\ref{update2}) the enhancement of 
the estimate $\xvec_{k+1}$ compared to  
$\xvec_{k}$ reads with respect to the measurement
$\yvec_{k+1}\in\mathbb{C}^{m+1}$ of a lowered $\ell_1$ norm
\begin{equation}\label{improvement3}
\big(C_{k+1}\xvec_{k+1} -\yvec_{k+1}\big)=
\left[R-C_{k+1}P_{k+1} C_{k+1}^{H}\right] 
R^{-1}\cdot\big(C_{k+1}\xvec_{k} -\yvec_{k+1}\big)\quad .
\end{equation}

Using the covariance cycle (\ref{ccycle2}) again the enhancement 
matrix composed from positive definite matrices $R$, $P$ and $Q$
on the RHS can 
be recast into the more suitable\footnote{a completely factorized 
version of (\ref{improvement2}) reads for $m< n$
\begin{equation}
R\cdot\left[R+C_{k+1}(P_{k}+Q) C_{k+1}^{H}\right]^{-1}
=\big(C_{k+1}(P_{k}+Q)C_{k+1}^H R^{-1}\big)^{-1}
\big(C_{k+1}P_{k+1}C_{k+1}^HR^{-1}\big)
\end{equation}}
form
\begin{equation}\label{improvement2}
\left[R-C_{k+1}P_{k+1} C_{k+1}^{H}\right] 
R^{-1} =R\cdot\left[R+C_{k+1}(P_{k}+Q) C_{k+1}^{H}\right]^{-1} 
\leq \mathbbm{1}_{m+1}
\end{equation}
yielding the inequality 
\begin{equation}\label{estimation}
\big|\langle\evec_j|C_{k+1}\xvec_{k+1} -\yvec_{k+1}\rangle\big| \leq
\big|\langle\evec_j|C_{k+1}\xvec_{k} -\yvec_{k+1}\rangle\big|\quad ,\quad
j=1,\ldots,m+1
\end{equation}
by components. The 
key ingredients are the restricted positive eigenvalues of 
(\ref{improvement2}) bounded from above by the spectral norm
$\norm{\mathbbm{1}_{m+1}}_2=1$, cf.\ (\ref{RQestimate}). 
What remains to prove for 
the inequality in (\ref{improvement2}) is the combination 
$C_{k}(P+Q)C_{k}^H$ to be positive semi definite: 

In the overdetermined case $m\geq n$ this holds true for 
$C_{k}^HC_{k}$ to be positive definite: As the combination 
$C_k(P+Q)C_k^H\in\mathbb{C}^{(m+1)\times(m+1)}$ has only rank~$n$ 
the amount of $m+1-n$ of its eigenvalues are zero. The positive 
sign of the remaining 
eigenvalues can be obtained from a Cholesky factorization of 
$P+Q={FF^H}$ with a lower triangular matrix 
${F}\in\mathbb{C}^{n\times n}$ yielding 
\mbox{$C_k {FF^H} C_k^H=(C_k{F})(C_k{F})^H$}. 
Due to a singular value decomposition \cite{GoLo13} of $C_k{F}$ 
the remaining eigenvalues belong to
$(C_k{F})^H(C_k{F})={F^H}(C_k^HC_k){F}\in\mathbb{C}^{n\times n}$ 
which, however, is along with the prerequisite $C_{k}^HC_{k}>0$ 
positive definite because of
\begin{theoremGoLo}
If $B\in\mathbb{C}^{n\times n}$ is positive definite and 
$X \in\mathbb{C}^{n\times k}$ has rank $k$, 
then the matrix $Y=X^HBX\in\mathbb{C}^{k\times k}$ is also 
positive definite.
\end{theoremGoLo}
By virtue of the same theorem the combination 
$C_k(P+Q)C_k^H\in\mathbb{C}^{(m+1)\times(m+1)}$ 
is always positive definite in the underdetermined case 
$m<n$, which is independent from the sensing matrix 
$C_k$.

So in each iteration step with adaptive \mbox{$0<\gamma_k\leq1$} 
the linearized $\ell_1$ norm is lowered (or remains at 
least constant) subject to the approximated constraint 
$\yvec=C_{k}\xvec_{k}$.
In the real case (\ref{modexpand}) even holds with \lq$=$\rq\
yielding for $j=m+1$ the relation 
\mbox{$\big|\norm{\xvec_{k+1}}_1-\gamma_k\norm{\xvec_k}_1\big|\leq(1-\gamma_k)\norm{\xvec_k}_1$}
implying \mbox{$\norm{\xvec_{k+1}}_1\leq\norm{\xvec_{k}}_1$}
for all $\gamma_k$.
To show a general convergence let us start with the exact 
solution $\xvec_\infty=\xvec$ and $C_\infty$ related to it. 
Because of $A=G=\mathbbm{1}$ with zero shifts $\uvec_k=\Nullvec$ 
in the Kalman filter model (\ref{themodel}) the constant (pseudo) 
measurement $\yvec_\infty$ already implies with $\gamma_k\equiv1$ 
a linear convergence of the series 
$\{C_\infty\xvec_0,C_\infty\xvec_1,C_\infty\xvec_2,\ldots\}$ by
\begin{equation} \label{weakconv}
\big|\langle\evec_j|C_\infty\xvec_{k+1} -\yvec_\infty\rangle\big| \leq
\big|\langle\evec_j|C_\infty\xvec_{k} -\yvec_\infty\rangle\big|
\leq\ldots\leq
\big|\langle\evec_j|C_\infty\xvec_{0} -\yvec_\infty\rangle\big|
\end{equation}
for all $j=1,\ldots,m+1$ starting with an $\xvec_0$ in the vicinity of 
$\xvec_\infty$. Thus the series $\big\{\xvec_0,\xvec_1,\xvec_2,\ldots\big\}$
is said to be {weakly} convergent and reconstructs the original signal
as long as the requirement for sparsity in the CS approach are met. 
In the framework of weak convergence the limiting value of $\Cov{\xvec_k}$ 
read for large $k\gg1$
\begin{equation}
P_k \geq
\begin{cases}
P_0 - P_0 C_\infty^H \big(C_\infty P_0 C_\infty^H\big)^{-1}C_\infty P_0
+\mathcal{O}\big(\frac1k\big) &\text{ for } m<n\\
\frac1k \big(C_\infty^H R^{-1}C_\infty\big)^{-1} 
+\mathcal{O}\big(\frac1{k^2}\big)
&\text{ for } m\geq n
\end{cases}
\end{equation}
where \lq$=$\rq\ represents the special case $Q\equiv0$. The limit 
$Q\to\infty$ is only possible for $m\geq n$ yielding a constant covariance
$P_k=(C_\infty^H R^{-1}C_\infty)^{-1}$ in each iteration.

\vspace*{1ex}
For an example of linear compressive sensing in the framework of Kalman 
filtering see \cite{LoEsCo15} using random samples of Fourier 
coefficients. 
Applying the method to a quadratic non-linear observation model is shown
for coherent $X$\!-ray scattering in the next Section~\ref{linobsmod}.

\section{Linearized Observation Model} \label{linobsmod}
Like the $\ell_1$ norm in subsection \ref{minimization} we linearize the 
squared observations (\ref{signal}) to apply the Kalman filtering scheme. 
Differentiating with respect to $x_k$ and 
$\overline{x_k}$ respectively yields
\begin{equation}
\frac{\partial}{\partial x_k}\big|S_r\big|^2 
=\Big(\langle\xvec|\!|T_r\Big)_k
\quad,\quad
\frac{\partial}{\partial \overline{x_k}}\big|S_r\big|^2 
=\Big(T_r|\xvec\rangle\Big)_k \quad
\end{equation}
and Taylor expansion around $\zvec_0\in\mathbb{C}^n$ reads up to 
1st order
\begin{equation}
\big|S_r\big|^2 = 2\Re \langle\zvec_0|\!|T_r|\xvec\rangle -
  \langle\zvec_0|\!|T_r|\zvec_0\rangle + \ldots\quad , 
\quad r=0,1,\ldots,n-1 \quad .
\end{equation}

With the biases $s_r=\langle\zvec_0|\!|T_r|\zvec_0\rangle$, 
$r=0,1,\ldots,n-1$ from the linearization the observation model is
\begin{align}\label{observationmodel}
\begin{pmatrix}\big|S_0\big|^2\\ \vdots\\ \big|S_{n-1}\big|^2\\
               \norm{\xvec}_1 \end{pmatrix} &=
\begin{pmatrix}2\langle\zvec_0|\!|T_0 \\ \vdots \\2\langle\zvec_0|\!|T_{n-1} \\ 
               \langle\pvec|(\zvec_0)\end{pmatrix}\cdot
\begin{pmatrix}x_0\\\vdots\\x_{n-2}\\x_{n-1}\end{pmatrix}
-\begin{pmatrix}s_0\\\vdots\\s_{n-1}\\0\end{pmatrix} \\[1em]
\yvec_k\hspace{1.9em} &=\hspace*{3em} C_k \hspace*{1.95em}\cdot\hspace*{1.6em}\xvec
\hspace*{1.5em} -\hspace*{1.4em}\svec_k
\end{align}
where the real part over all $C\xvec$ combinations  have to be 
taken. The model in detail reads 
$\yvec_k=\Re\big(C_k\xvec\big)-\svec_k\in\mathbb{R}^{n+1}$ and 
because of the $C\xvec$ combinations we choose the alternative 
representation (\ref{alternative}) of the Kalman filter equations
to invert (\ref{observationmodel}). Then the 
iterated estimates $\xvec_k$ to the state vector $\xvec$ read
\begin{subequations}\label{filtereq}
\begin{align}
K_{k+1}&=(P_{k}+Q) C_{k+1}^H 
  \big(C_{k+1}(P_{k}+Q) C_{k+1}^H+R\big)^{-1}\\
P_{k+1} &= \big(\mathbbm{1}_n-K_{k+1}C_{k+1}\big)(P_{k}+Q)\\
\xvec_{k+1} &= \xvec_{k} + 
K_{k+1}\cdot\Re\Big[\yvec_{k+1}-C_{k+1}\xvec_{k}+\svec_{k+1}\Big]
\end{align}
\end{subequations}
with the complex sensing matrix $C_{k+1}\in \mathbb{C}^{(n+1)\times n}$, 
biases $\svec_{k+1}\in\mathbb{C}^{n+1}$ and the given real (pseudo) 
measurements $\yvec_{k+1}\in\mathbb{R}^{n+1}$ according to
\begin{equation}
\begin{aligned}
C_{k+1}=\begin{pmatrix}2\langle\xvec_k|\!|T_0 \\ 
   \vdots \\2\langle\xvec_k|\!|T_{n-1} \\ 
   \langle\pvec|(\xvec_k)\end{pmatrix} 
\end{aligned} \;\; ,\;\;
\begin{aligned}
\svec_{k+1}=
  \begin{pmatrix}\langle\xvec_k|\!|T_0|\xvec_k\rangle\\ \vdots\\
\langle\xvec_k|\!|T_{n-1}|\xvec_k\rangle\\ 0
\end{pmatrix}
\end{aligned} \;\; ,\;\;
\begin{aligned}
\yvec_{k+1}=\begin{pmatrix}\big|S_0\big|^2\\ 
  \vdots\\ \big|S_{n-1}\big|^2\\
  \gamma_{k}\norm{\xvec_k}_1 \end{pmatrix} 
\end{aligned}
\;\; .
\end{equation}
\vspace*{1ex}

A properly chosen factor $0<\gamma_{k}\leq1$ adaptively lowers the 
$\ell_1$ norm in each iteration step~$k$ to reconstruct the signal 
$\xvec=\lim_{k\to\infty}\xvec_k$ as a limiting value from its 
given squared measurements 
$\langle\yvec|=\big(\big|S_0\big|^2,\ldots,\big|S_{n-1}\big|^2\big)$. 
Note that $\{T_0|\xvec_k\rangle,T_1|\xvec_k\rangle,
\ldots,T_{n-1}|\xvec_k\rangle\}\subset\mathbb{C}^n$ form an orthogonal 
basis. Thus
\begin{equation}
C_{k+1}^H C_{k+1} 
= 4\sum_{r=0}^{n-1}T_r|\xvec_k\rangle\langle\xvec_k|\!|T_r 
+ |\pvec\rangle\langle\pvec|(\xvec_k)
\end{equation}
is hermitian and positive definite due to (\ref{RQestimate}) for all 
$k$ with $|\xvec_k\rangle\not=|\Nullvec\rangle$ which is sufficient 
to prove weak convergence of the $\ell_1$ minimization under the 
constraint of squared observations, cf.\ 
(\ref{improvement2})-(\ref{weakconv}).

\subsection{Example 1: Stacking Sequence}

\begin{figure}
\includegraphics[width=0.8\textwidth]{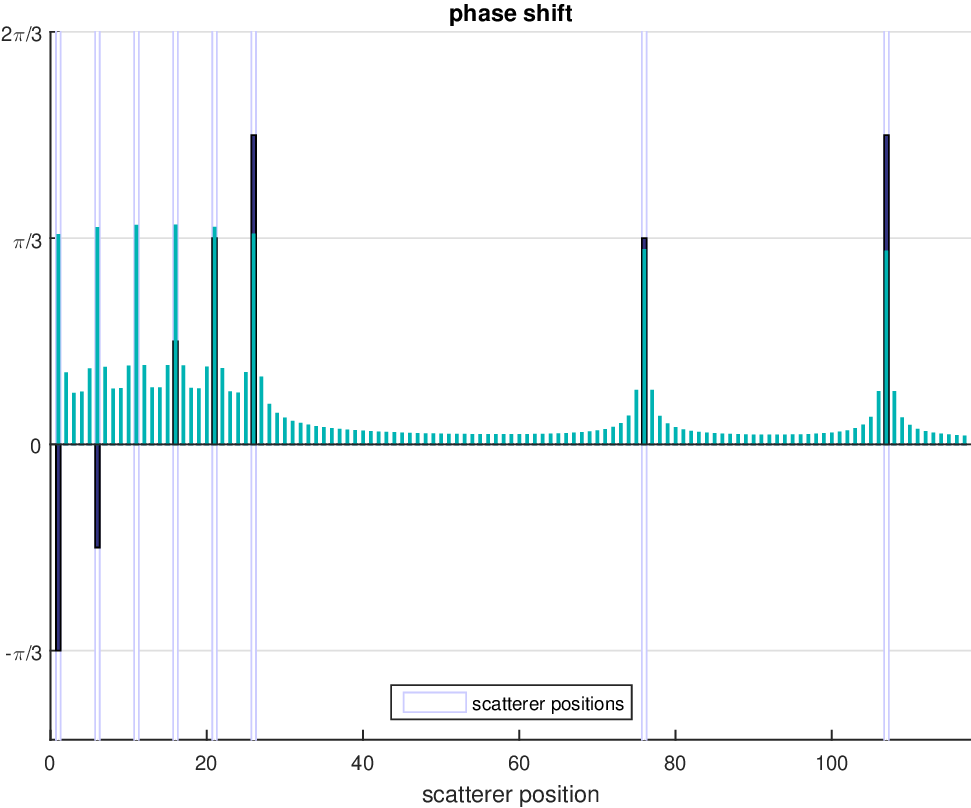}
\vspace*{2em}

\hspace*{0.01ex}
\includegraphics[width=0.78\textwidth]{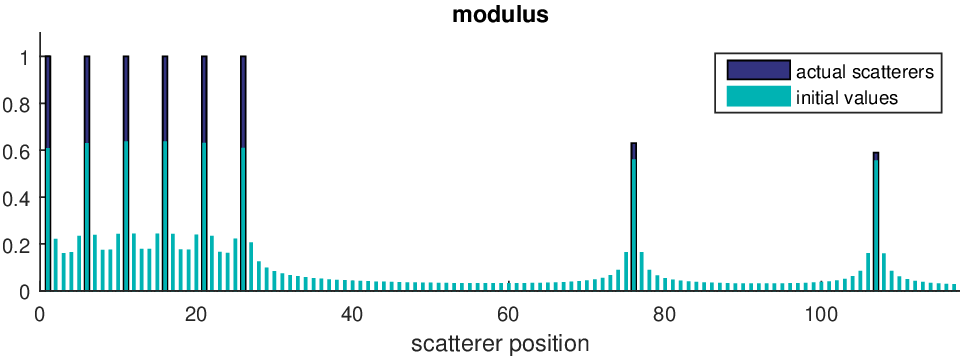}
\caption
{A number of $6$ equidistant scatterers on a periodic chain of length 
$117$. To allow for a signal from the substrate the wires are growing 
on we added some parasitic scatterers to the chain. As initial data we 
use the moduli and phases from $\xvec_{\mathrm{sparse}}$ symmetrically 
broadened by a leakage $\frac12(\fvec_0^+ + \fvec_0^-)$, cf.\ 
(\ref{leakage}), with shift $s=0.3$ and amplitudes $0.6$ and $1.1$ 
respectively.} 
\label{settingstack}
\end{figure}

To demonstrate the relative phase recovery from the observed intensity
(\ref{signal}) not only restricting to the Zinc-blende and Wurtzite phases 
(\ref{phasesc}) assume a linear chain with $117$~lattice sites and 
$6$~equidistant sparse amplitudes $|x_j|=1$ with phases from the set 
$\{-\frac{\pi}{3},-\frac{\pi}{6},0,\frac{\pi}{6},\frac{\pi}{3},\frac{\pi}{2}\}$ 
forming $\xvec_{\mathrm{sparse}}\in\mathbb{C}^n$. 

\begin{figure}
\includegraphics[width=0.83\textwidth]{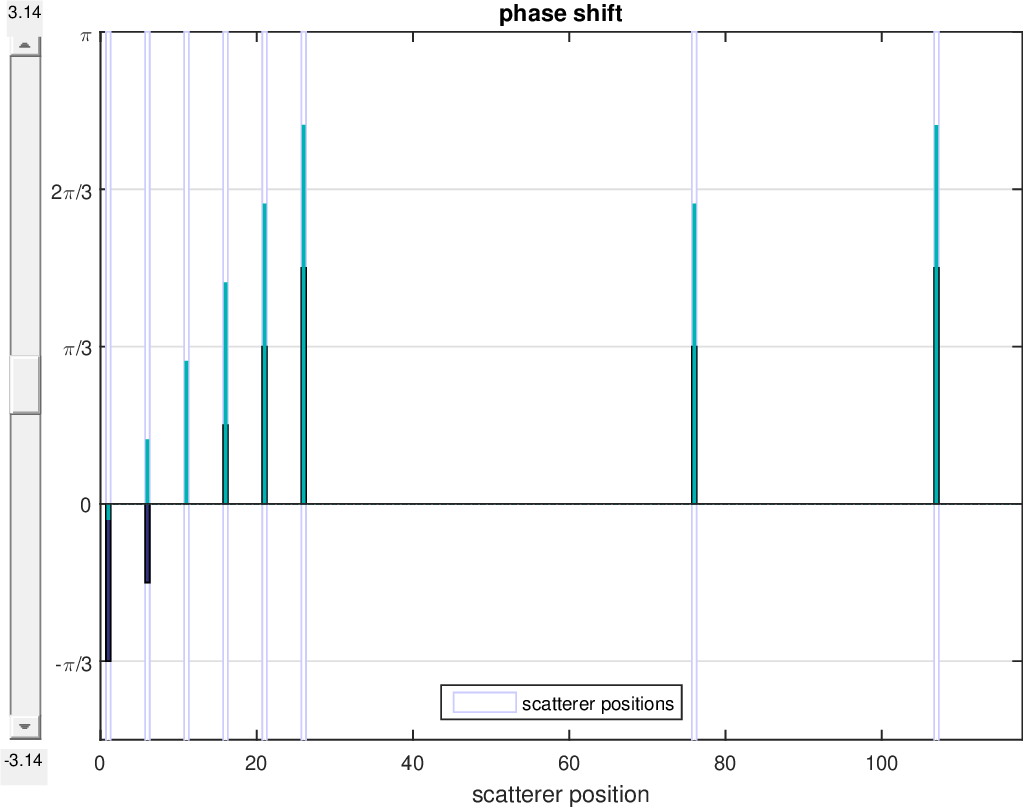}
\vspace*{2em}

\hspace*{1.4em}
\includegraphics[width=0.78\textwidth]{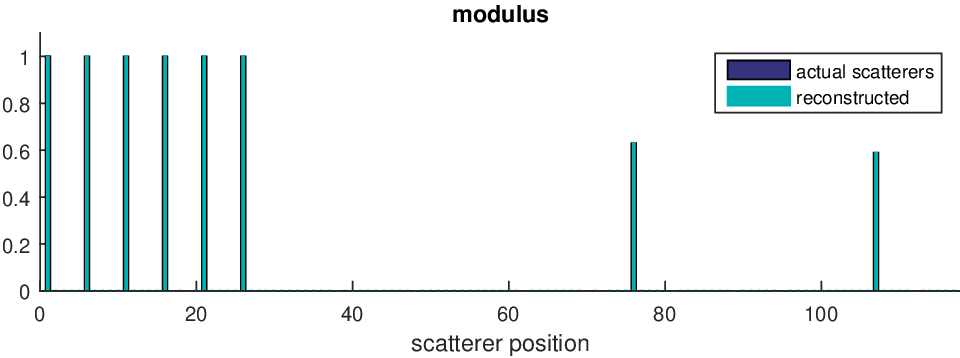}
\caption
{Reconstructed moduli of all scatterers. In the phase domain the values
corresponding to vanishing moduli are suppressed. Due to the symmetry of
the sensing problem the phases are only reconstructed up to a
global phase.}
\label{resultstack}
\end{figure}

\begin{figure}
\includegraphics[width=0.83\textwidth]{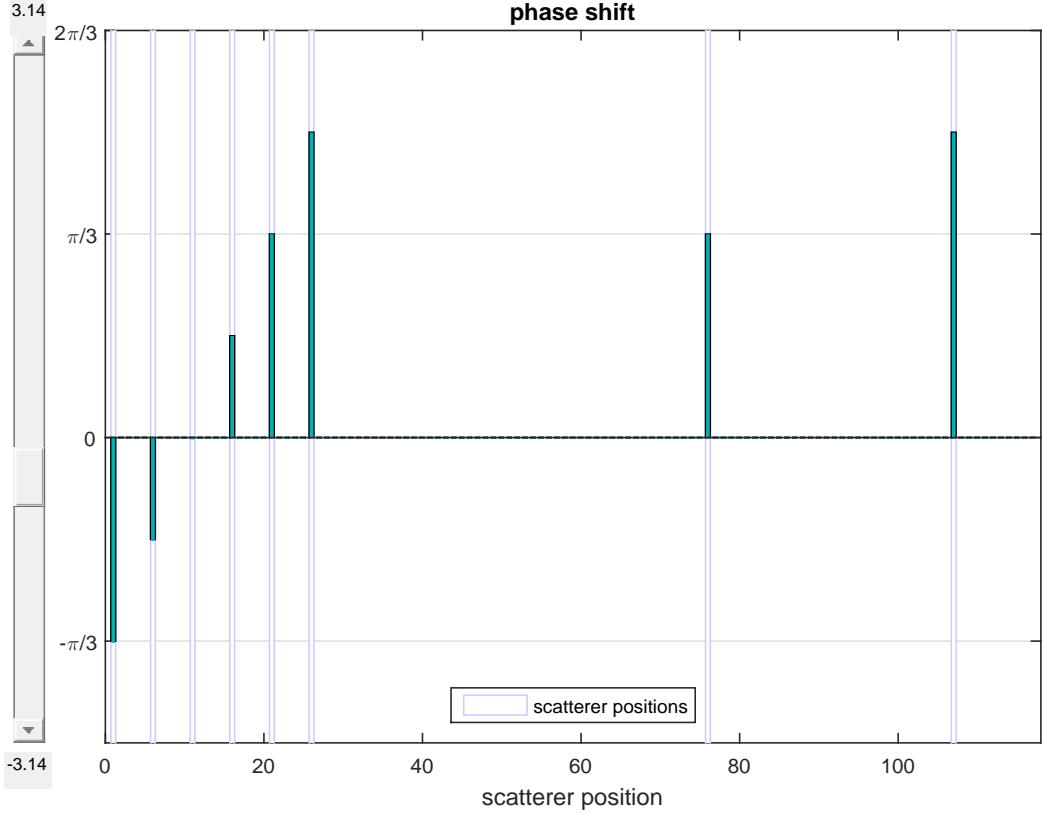}
\caption
{Adding a global phase (slider on the left) verifies that all 
relative phases are reconstructed correctly. The open rectangles 
mark the position of occupied lattice sites on the chain. In the 
original objective this would refer to bilayers and substrate in 
the nanowire.} \label{resultstackslider}
\end{figure}

\vspace*{1em}
As in the original objective the number of bilayers is roughly known 
we use as initial values $\xvec_0\in\mathbb{C}^n$ a broadened modulus 
and phase distribution related to the setting $\xvec_{\mathrm{sparse}}$,
cf.\ Figure~\ref{settingstack}: With the support 
$\mathbb{S}:=\supp \xvec_{\mathrm{sparse}}$ a leakage
\begin{equation} \label{leakage}
\big(\fvec^\pm_0\big)_k =\sum_{\alpha\in \mathbb{S}} 
\frac{\sin\big((k\pm s-\alpha)\pi\big)}{(k\pm s-\alpha)\pi}\in\mathbb{R}
\quad , \quad k=0,\ldots,n-1 
\end{equation}
for the moduli and phases of $\xvec_{\mathrm{sparse}}$ is with 
shifts $0\leq s <1$ a suitable way to also consider the limit $s\to0$ of 
exactly known positions. Note that for the phases we only used
one fixed value for each occupied lattice site broadened by the 
leakage. So the main effort is the recovery of the relative phases
from the measurements.

\begin{figure}
\includegraphics[width=0.78\textwidth]{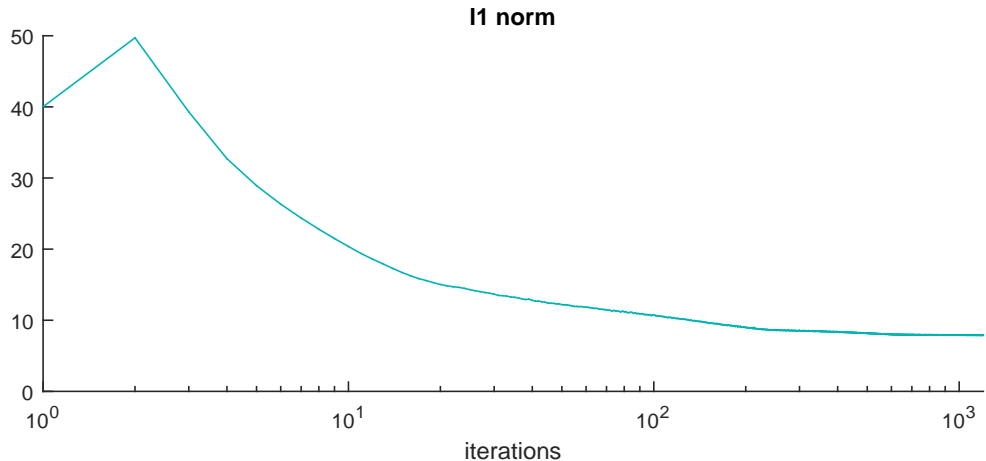}
\caption
{Note that the $\ell_1$ convergence (\ref{estimation}) in the $k$th
iteration step relates the linearized norm 
$\langle\pvec|\xvec_k\rangle \leq \norm{\xvec_{k}}_1$
to $\norm{\xvec_{k-1}}_1$. 
Because of this inequality, cf.\ 
(\ref{normbound}), there is no guarantee of a monotonically lowered 
$\ell_1$~norm, which can be seen in the beginning (iteration step $2$).} 
\label{l1converge}
\end{figure}

\vspace*{1ex}
For an exponential decaying lowering factor $\gamma_k=1-0.1 \exp(-0.0019 k)$ 
within $1200$ iterations the reconstruction results for noiseless
synthetic measurements are shown in 
Figures~\ref{resultstack}~and~\ref{resultstackslider}.
Experientially, good algorithm's covariances were found to be
\begin{equation}
P_0=0.3\cdot\mathbbm{1}_{117} \;,\;
Q=10^{-8}\cdot\mathbbm{1}_{117}\;,\;
R=\diag(10^{-4},\ldots,10^{-4},10^{-6})\in\mathbb{R}^{118\times118}
\end{equation}
\newpage

To drive the $\ell_1$ minimization the corresponding entry 
(here $10^{-6}$) in the $R$ matrix has to be much lower in magnitude 
than the ones (here $10^{-4}$) related to the observations of the 
constraint. When reconstructed amplitudes are said to vanish (e.g.\ 
in Figure~\ref{resultstack}) this means only up to numerical 
precision where the zero value is dominated by the entries of $Q$.
Empirically, the exponential decay in the lowering factor $\gamma_k$ 
is chosen such that its values approaches about $0.99,\ldots,0.999$ 
after the maximum number of iterations. A resulting typical smooth 
convergence of the $\ell_1$~norm due to the exponential decay is 
shown in Figure~\ref{l1converge}: {The constant tail
can be used to determine a maximal number of iterations as a stopping 
criterion. Note that the reconstructed results beyond this guessed 
number are insensitive to further iterations meaning for the 
equations (\ref{filtereq}) to perform fixed point iterations within 
accuracies related to the covariances $Q$ and $P_k$.} 
\vspace*{1em}

In the nanowire the equidistant bilayers are at fixed positions in 
growth direction $\avec_3$. Thus the reconstruction of their complex 
scattering amplitudes does not suffer from \lq off grid\rq\ problems
in general, see e.g.\ \cite{CSchPC11,DuBa13,TBSR13}, if integer multiples
of the lattice constants are used for the DFT.

\subsection{Example 2: Pattern Reconstruction}

The 2D reconstruction from (\ref{FTsquare2D}) can be calculated with 
the same algorithm already used for Example~1 above, as the 2D data 
set of size $12\times25$ can be mapped to a 1D~vector by means of 
(\ref{mapping2D1D}), cf.\ Figure~\ref{initial2D}. Like the example 
from Figure~\ref{settingstack}, we use as initial values with the 
same reasoning and broadening parameters the given distribution.
Within $2000$ iterations and an exponentially decaying lowering factor 
\mbox{$\gamma_k=1-0.17\cdot\exp(-0.0028 k)$} the reconstruction for
noiseless measurements is shown in Figure~\ref{nicereconstruct} 
with empirically good algorithm's covariances
\begin{equation}
P_0=0.3\cdot\mathbbm{1}_{300} \;,\;
Q=10^{-7}\cdot\mathbbm{1}_{300}\;,\;
R=\diag(10^{-4},\ldots,10^{-4},10^{-6})\in\mathbb{R}^{301\times301}
\;\;.
\end{equation}
Note that there is a combined translational and reflectional 
symmetry mapping the vectorized real scattering amplitude
distribution onto itself, cf.\ Figure~\ref{initial2D}. As this 
is also a symmetry of the squared observations (\ref{signal})
in the linear configuration, 
there are two competing solutions to the sensing problem. 
Interestingly, breaking this symmetry improved the 
$\ell_1$~convergence in Figure~\ref{nicereconstruct}:
Without the parasitic scatterers the symmetric positions of
consecutive sites\footnote{At least in linear CS with a 
constant sensing matrix there is also to consider the 
resolution of consecutive frequency bins in discrete Fourier 
transform, cf.\ \cite{TBSR13} and references therein.} are
correctly {(with respect to the support)} 
found after $5000$ iterations, cf.\ Figure~\ref{artefact}, 
and the remaining discrepancies in the amplitudes do not even 
vanish completely after inefficient $50000$ iterations. 
On the contrary the random sample of $25$ scatterers from 
Figure~\ref{random2D} is already reconstructed with less than 
$1000$ iterations suggesting a strong dependence of the 
algorithm's convergence from the vector under consideration. 
\begin{figure}
\includegraphics[height=0.4\textheight]{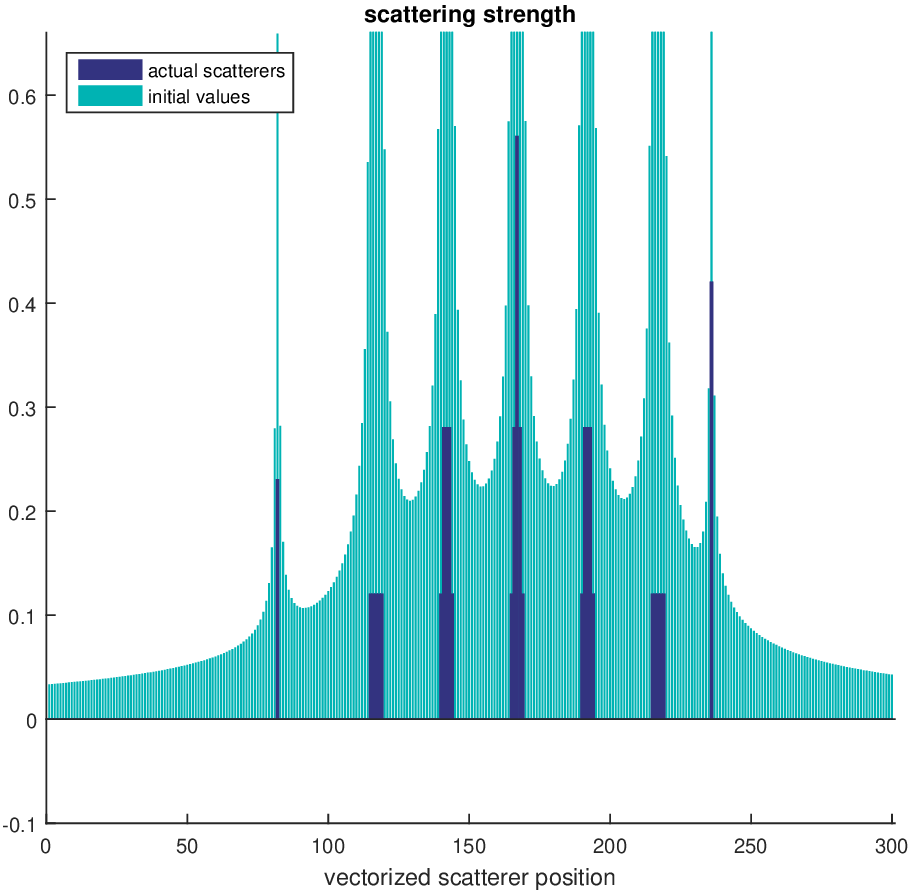}
\includegraphics[height=0.4\textheight]{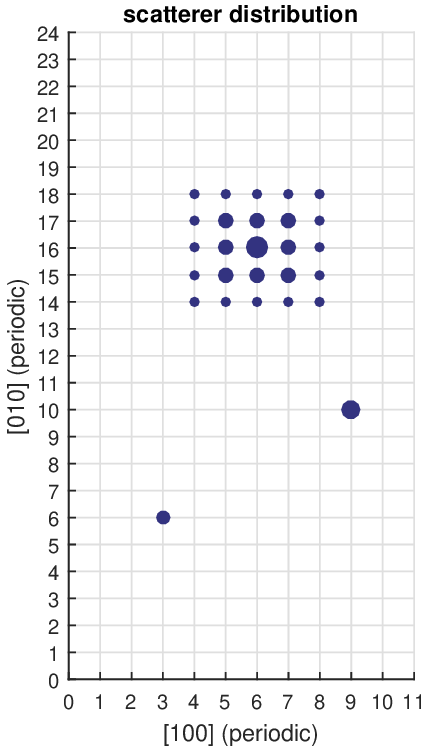}
\caption
{The actual scatterers forming the $5\times5$ structure are made
from real amplitudes $0.12$, $0.28$ and $0.56$ whereas the two 
parasitic scatterers are assigned with $0.23$ and $0.42$ (left 
panel). The scattering strength in the plane is represented by 
the area of the filled circles (right panel).} 
\label{initial2D}
\end{figure}
\begin{figure}
\includegraphics[height=0.4\textheight]{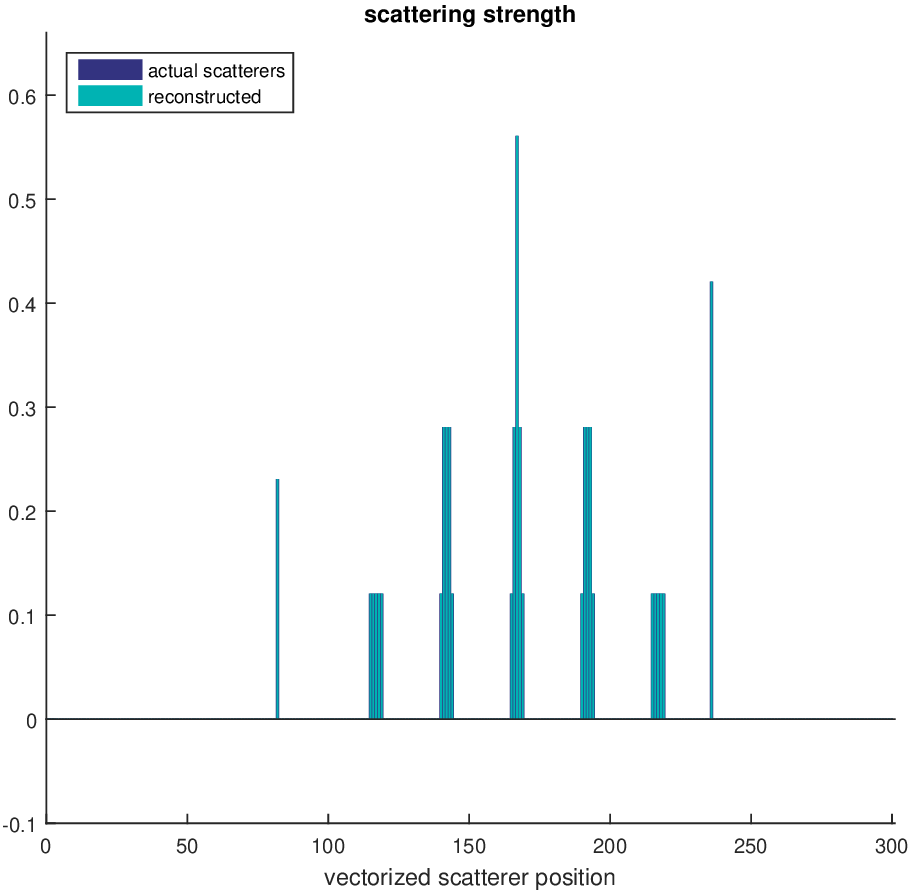}
\includegraphics[height=0.4\textheight]{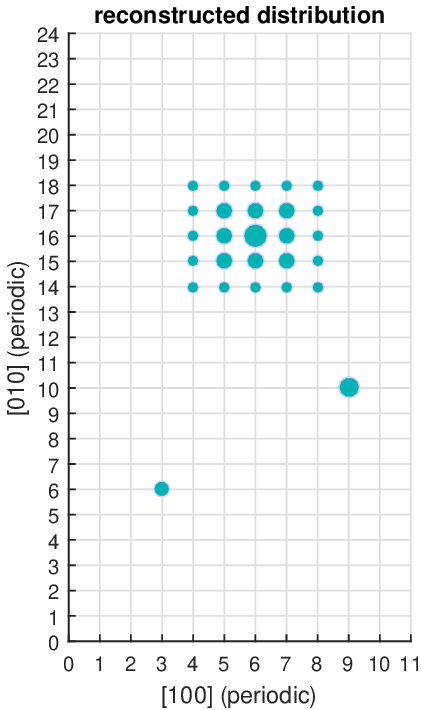}
\caption
{The real amplitudes are nicely reconstructed (left panel) after 
$2000$ iterations which can also be seen in the 2D setting (right 
panel) by comparing with Figure~\ref{initial2D}.} 
\label{nicereconstruct}
\end{figure}
\begin{figure}
\includegraphics[height=0.394\textheight]{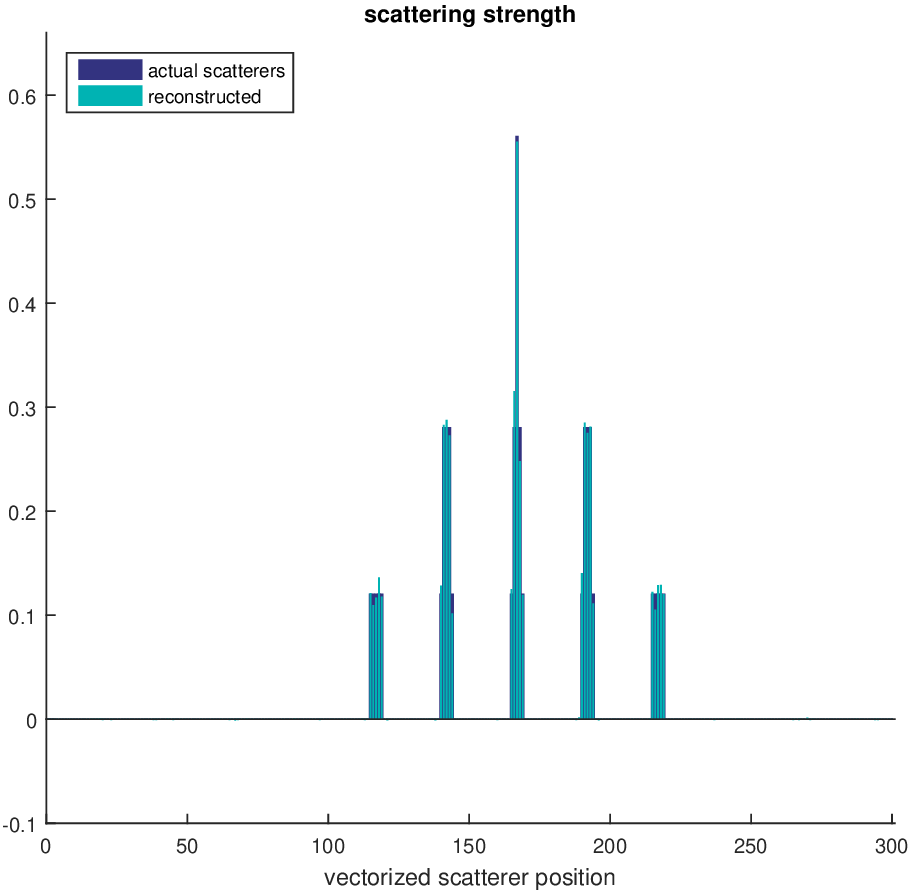}
\includegraphics[height=0.394\textheight]{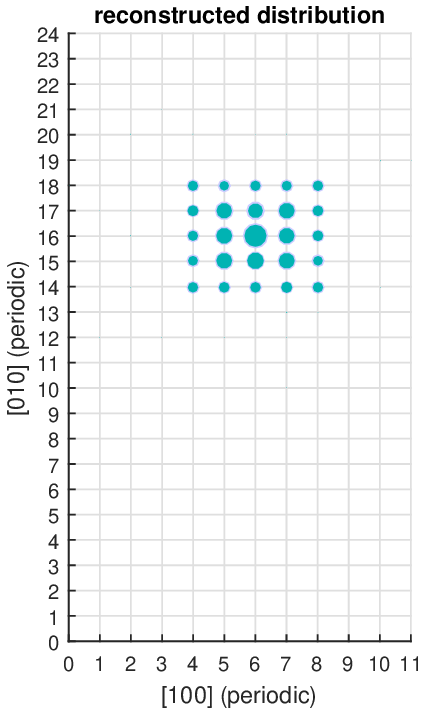}
\caption
{Without parasitic scatterers the positions are correctly found 
(right panel) after $5000$ iterations and 
$\gamma_k=1-0.17\cdot\exp(-0.0012 k)$. The remaining discrepancies
in the amplitudes seem to be complementarily symmetric with respect 
to the central peak (left panel).} \label{artefact}
\end{figure}
\begin{figure}
\includegraphics[height=0.39\textheight]{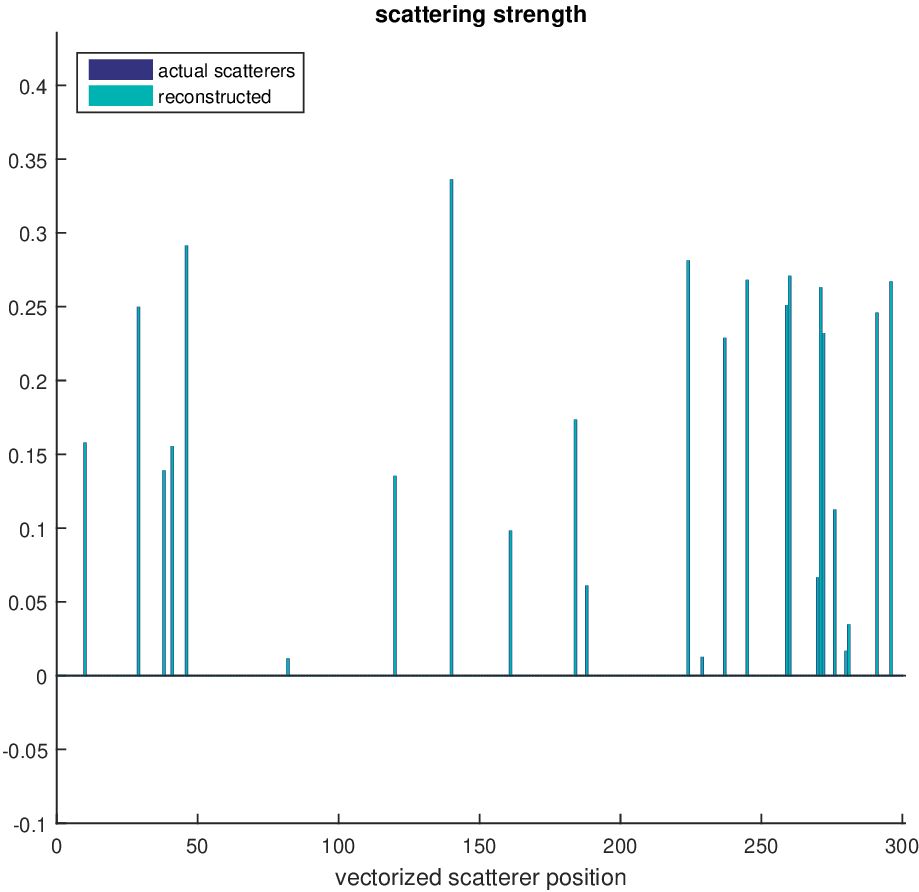}
\includegraphics[height=0.39\textheight]{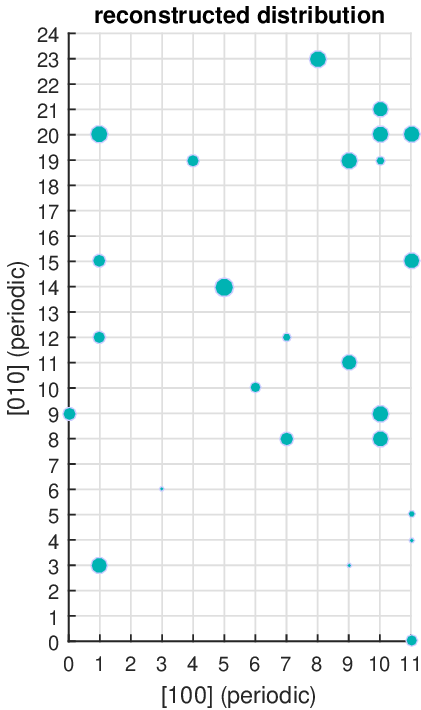}
\caption
{A random distribution of $25$ scatterers with real random 
amplitudes normalized in the $\ell_2$~norm. The reconstruction is 
without artefacts and discrepancies reached after $1000$ iterations 
with $\gamma_k=1-0.17\cdot\exp(-0.0058 k)$.} 
\label{random2D}
\end{figure}

\section{Conclusion and Outlook} \label{outlook}

We aimed to implement non-linear compressive sensing for complex 
vectors $\xvec$ by inverting the underlying observation model with 
Kalman filtering. For this reason we proved a weak convergence of 
the filter equations for complex sensing matrices $C$ and hermitian 
covariances~$R,P,Q$. For the example of quadratic nonlinearities we 
applied our formulas to retrieve relative phase information in the 
objective of simulated noiseless coherent $X$\!-ray diffraction. Due to 
the nonlinearity the noise in the intensity is not Gaussian and needs 
further considerations to also account for e.g.\ lattice distortions or the
finite coherence of the primary beam. As an outlook we presented a 2D 
pattern reconstruction which, hopefully, can help to investigate if 
the cross sections of the wires were grown regularly.
\enlargethispage{1\baselineskip}

Because of the Jacobians building up the sensing matrix $C$ the 
convergence of the underlying $\ell_1$ norm minimization does depend 
on the reconstructed vector $\xvec$ and may go beyond the resolution
issues \cite{TBSR13} for even constant sensing matrices. For this 
reason a thorough analysis on the relation between maximal sparsity
{(for the examples here about 10\% of the 
available lattice sites)} and the Toeplitz matrix (\ref{toeplitz}) 
representing the sensing process for a successful CS is needed. 
As mentioned in the introduction algorithms in general minimizing 
the $\ell_1$ norm could be of interest retrieving information out 
from intensity spectra: {To compare our results we 
applied the non-linear version \cite{Valkonen14} of the primal dual
algorithm \cite{ChPo10} to the 1D sensing problem above yielding similar
results with respect to iteration numbers, reconstructed amplitudes 
and phases. On the one hand primal dual reconstructed the zeros 
outside the support (once a solution was isolated from the algorithm) 
numerically exact compared to accuracies of orders $10^{-5}$ reached by 
the Kalman filter. On the other hand our approach focused in the first 
quarter of the iterations on guessing the support by strongly increasing 
and decreasing the corresponding amplitudes, whereas primal dual 
homogeneously acted on all the amplitudes during all iterations. For 
a detailed comparison and maybe a combination, the whole parameter 
range of the approaches have to be investigated.} 
To apply the Kalman filter-based algorithm above to recorded data 
\cite{DBLP16} we need to deal with several $10^4$ $\xvec$-vector's entries
describing a sparsely (by a factor of roughly~\mbox{$1:10^1$)} occupied 
linear grid covering all the approximately~$10^3$ bilayers in a nanowire 
of $500\,$nm in height. For this reason matrix inversions, 
cf.~(\ref{filtereq}), should be reformulated by e.g.\ sequential processing 
techniques to make the algorithm more scalable. Furthermore it would also 
be interesting if other adaptive lowering factors $\gamma_k$ compared 
to the exponential ones yield faster convergences.

As the reconstruction from quadratic constraints seems to depend only
little on the explicit algorithm used for the $\ell_1$ minimization, 
it would be interesting to figure out to which extend the linear matrix 
completion approach \cite{CESV13,ChMoPa10,FoRaWa11} with respect to the 
nuclear norm can still be applied in face of seemingly too few numbers 
of independent observations.

\begin{acknowledgments}
This work has been supported in part by the Deutsche Forschungsgemeinschaft
under grants Lo~455/20-1 and Pi~214/38-2.
\end{acknowledgments}

\begin{appendix}

\section{Vector Notations for Complex Numbers} \label{appvectors}
Vectors over the field $\mathbb{C}$ consist of complex 
numbers $z=x+\I y$ with $x,y\in\mathbb{R}$ and $\I^2=-1$. 
With a vector $|\vvec\rangle\in\mathbb{C}^n$ we usually associate $n$ 
complex numbers $v_1,v_2,\ldots,v_n$ arranged as a column.
With the complex
conjugate $\overline{z}:=x-\I y$ row vectors $\langle\vvec|\in\mathbb{C}^n$ 
are considered to be dual to $|\vvec\rangle$ by
\begin{equation}
\begin{gathered}
|\vvec\rangle=\begin{pmatrix}v_1\\\vdots\\v_n\end{pmatrix} \in \mathbb{C}^n
\end{gathered}
\quad ,\quad
\begin{aligned}
\langle\vvec|:&= \big(\overline{v_1}, \ldots ,
\overline{v_n}\big) \in \mathbb{C}^n\\
|\vvec\rangle&=\sum_{k=1}^n v_k |\evec_k\rangle \;\; , \;\; 
\big(|\vvec\rangle\big)_j=v_j\in\mathbb{C}
\end{aligned}
\end{equation}
where $\big(|\evec_j\rangle\big)_k =\delta_{jk}$ denotes with 
$j,k=1,\ldots,n$ the usual standard basis. With the complex scalar 
product
\begin{equation}
\langle\vvec|\wvec\rangle=\sum_{k=1}^n \overline{v_k}w_k \quad ,\quad 
|\vvec\rangle,|\wvec\rangle\in\mathbb{C}^n
\end{equation}
each vector $|\vvec\rangle\in\mathbb{C}^n$ can be represented in an 
unitary basis 
$\big\{|\qvec_1\rangle,\ldots,|\qvec_n\rangle\big\}\subset\mathbb{C}^n$ 
yielding the expansion
\begin{equation}
|\vvec\rangle = \sum_{k=1}^n|\qvec_k\rangle\langle\qvec_k|\vvec\rangle 
\quad , \quad
\langle \qvec_j|\qvec_k\rangle=\delta_{jk}\text{ with } 
j,k=1,\ldots,n \quad .
\end{equation}

With $A=(a_{ij})_{ij}\in\mathbb{C}^{m\times n}$ we denote matrices 
with $m$ rows and $n$ columns consisting of the complex entries 
$a_{ij}=(A)_{ij}$. Let $(A^H)_{ij}:=\overline{(A)_{ji}}=\overline{a_{ji}}$ 
represent the hermitian conjugate then we get from its column decomposition
$A^H=\big(|a_1\rangle,\ldots,|a_n\rangle\big)\in\mathbb{C}^{m\times n}$
\begin{equation}
A=\begin{pmatrix}\langle a_1|\\ \vdots\\ \langle a_n|\end{pmatrix} 
\quad ,\quad
\langle a_j|=|a_j\rangle^H\in\mathbb{C}^m\quad .
\end{equation}
Using such a row decomposition the matrix multiplication of 
commensurable $A\in\mathbb{C}^{m\times n}$ and 
$B\in\mathbb{C}^{n\times p}$ can be viewed as $mp$ single scalar 
products yielding with 
$B=\big(|\bvec_1\rangle,\ldots,|\bvec_p\rangle\big)$ 
\begin{equation}
AB=C=(c_{ik})_{ik}\in\mathbb{C}^{m\times p} \quad ,\quad
c_{ik} = \sum_{\ell=1}^n a_{i\ell}b_{\ell k} = 
\langle\avec_i|\bvec_k\rangle
\end{equation}
including products $A|\xvec\rangle\in\mathbb{C}^m$ for 
$B=|\xvec\rangle\in\mathbb{C}^{n}$. Thus with the identity 
$\mathbbm{1}_n=(\delta_{ij})_{ij}\in\mathbb{C}^{n\times n}$ 
unitary matrices 
$Q=\big(|\qvec_1\rangle,\ldots,|\qvec_n\rangle\big)\in\mathbb{C}^{n\times n}$ 
consist of orthonormal (row and) columns defined by the property
\begin{equation}
Q Q^H = Q^H Q = \mathbbm{1}_n \quad , \quad 
\big(Q^H Q\big)_{jk} = \langle\qvec_j|\qvec_k\rangle = \delta_{jk}
\quad , \quad 
Q^H = Q^{-1} \quad .
\end{equation}

The possibility of multiplying matrices 
$A=(a_{ij})_{ij}\in\mathbb{C}^{m\times n}$ and 
$B=(b_{rs})_{rs}\in\mathbb{C}^{p\times q}$ of arbitrary 
dimensions is covered by the Kronecker product
\vspace*{1ex}
\begin{equation} \label{kronecker}
A\otimes B = 
\left(
\begin{array}{c|c|c}
  \begin{smallmatrix}
  a_{11}b_{11}  & \hdots  & a_{11}b_{1 q}\\[-.4ex]
  \vdots &\ddots& \vdots\\[.5ex]
  a_{11}b_{p1} & \hdots & a_{11}b_{pq}
  \end{smallmatrix}
& \hspace{1em}\hdots\hspace{1em} &
  \begin{smallmatrix}
  a_{1 n}b_{11}  & \hdots  & a_{1 n}b_{1 q}\\[-.4ex]
  \vdots &\ddots& \vdots\\[.5ex]
  a_{1 n}b_{p1} & \hdots & a_{1 n}b_{pq}
  \end{smallmatrix}
\\[2.5ex] \hline\vdots & \ddots  & \vdots\\ \hline \phantom{.} & & \phantom{.}\\[-3ex]
  \begin{smallmatrix}
  a_{m 1}b_{11}  & \hdots  & a_{m 1}b_{1 q}\\[-.4ex]
  \vdots &\ddots & \vdots\\[.5ex]
  a_{m 1}b_{p1} & \hdots & a_{m 1}b_{pq}
  \end{smallmatrix}
& \hdots &
  \begin{smallmatrix}
  a_{mn}b_{11}  & \hdots  & a_{mn}b_{1 q}\\[-.4ex]
  \vdots &\ddots& \vdots\\[.5ex]
  a_{mn}b_{p1} & \hdots & a_{mn}b_{pq}
  \end{smallmatrix}
\end{array}
\right) \in\mathbb{C}^{(mp)\times(nq)}
\vspace*{1ex}
\end{equation}
reading in components $\big(A\otimes B\big)_{(ir)(js)}=a_{ij} b_{rs}$
with multiple row and column indices $(ir)$ and $(js)$ respectively.
For $|\xvec\rangle\in\mathbb{C}^n$, $|\yvec\rangle\in\mathbb{C}^q$
and $\lambda\in\mathbb{C}$ the Kronecker product meets
\begin{equation}
\big(\lambda A|\xvec\rangle\big) \otimes \big(B|\yvec\rangle\big)
= \big(A|\xvec\rangle\big) \otimes \big(\lambda B|\yvec\rangle\big)
= \lambda\big(A\otimes B\big) \big(|\xvec\rangle\otimes|\yvec\rangle\big)
\quad.
\end{equation}

Introducing the length of row or column vectors $\xvec\in\mathbb{C}^n$ 
can be accomplished with the $\ell_p$ norms $\norm{\xvec}_p$ defined for 
all real $p>0$ by
\begin{equation}
\norm{\xvec}_p^p := \sum_{j=1}^n |x_j|^p \quad , \quad 
\norm{\xvec}_0:=\card\{j\,|\,x_j\not=0,j=1,\ldots,n\} \quad .
\end{equation}
Especially $\ell_0$ is no norm but can be viewed as 
the number of non-zero entries. The corresponding matrix norms 
of $A\in\mathbb{C}^{m\times n}$ can be related to the vector 
norms according to
\begin{equation}
\norm{A}_p := \sup_{\xvec\in\mathbb{C}^n} 
  \frac{\norm{A\xvec}_p}{\norm{\xvec}_p}
  =\sup_{\norm{\xvec}_p=1}\norm{A\xvec}_p \quad .
\end{equation}

\section{Hermitian Matrices}
Hermitian matrices $X,Y\in\mathbb{C}^{n\times n}$ can be related by
$X>Y$ ($X\geq Y$) if $X-Y$ is positive (semi) definite which can be 
written as $X-Y>0$ ($X-Y\geq 0$). In addition
\begin{theoremHB}
If $X,Y\in\mathbb{C}^{n\times n}$ are hermitian with $X\geq Y>0$.
Then $0< X^{-1} \leq Y^{-1}$.
\end{theoremHB}
\begin{proof}$Y=Y^{\frac12} Y^{\frac12}$, cf.\ \S~9.2.4 in \cite{GoLo13}
$\Rightarrow \mathbbm{1}_n \leq Y^{-\frac12} X Y^{-\frac12}$
has eigenvalues larger than $1$
$\Rightarrow Y^{\frac12} X^{-1} Y^{\frac12}$ has positive 
eigenvalues lower than $1$
$\Rightarrow 0<Y^{\frac12} X^{-1} Y^{\frac12}\leq \mathbbm{1}_n$.
\end{proof}

Thus all hermitian matrices $R,X\in\mathbb{C}^{n\times n}$ with  
$R$ positive definite and $X$ positive semi definite 
satisfy the inequalities
\begin{equation} \label{RQestimate}
R+X\geq R > 0 \quad \Leftrightarrow \quad 0 < \big(R+X\big)^{-1} \leq R^{-1}
\end{equation}
\end{appendix}

\bibliographystyle{amsplain}
\bibliography{cs,kalman,nano}

\end{document}